\documentclass[12pt]{article}

\usepackage{amsmath,amssymb,amsthm,color,epsfig,mathrsfs}
\usepackage{amsbsy,float}
\usepackage[normalem]{ulem}

\usepackage{tabularx,ragged2e,booktabs,caption}
\newtheorem{theorem}{Theorem}

\newcounter{spslist}
\newenvironment{spslist}{
  \begin{list}
  {\begin{picture}(1,1)
     \setlength{\unitlength}{0.5cm}
     \put(0,0.22){\circle*{0.2}}
    \end{picture}}
  {\usecounter{spslist}
  \setlength{\leftmargin}{1em}
  \setlength{\labelsep}{0.6em}
  \setlength{\labelwidth}{1em}
  \setlength{\topsep}{1ex}
  \setlength{\rightmargin}{0em}
  \setlength{\itemsep}{0.5ex}
  \setlength{\parsep}{0em}
  \setlength{\itemindent}{0em} }}
  {\end{list}}

\newcommand{\col}[3]{ \renewcommand{\arraystretch}{#1}
                \left[\!\! \begin{array}{c} #2 \\ #3 \end{array} \!\!\right] }

\newcommand{\mat}[5]{ \renewcommand{\arraystretch}{#1}
                    \left[\! \begin{array}{cc}
                            #2 & #3 \\
                            #4 & #5 \end{array} \!\right] }
\newcounter{geqncount}
    {\refstepcounter{equation}%
     \setcounter{geqncount}{\value{equation}}%
     \setcounter{equation}{0}%
  }%
    {\setcounter{equation}{\value{geqncount}}}
\newcommand{\aaa}{\alpha}
\newcommand{\bb}{\beta}
\newcommand{\g}{\gamma}
\newcommand{\z}{\zeta}
\newcommand{\asq}{\alpha^2}
\newcommand{\bsq}{\beta^2}
\newcommand{\gsq}{\gamma^2}

\newcommand{\ssq}{s^2}
\newcommand{\zsq}{\zeta^2}

\newcommand{\half}{{\textstyle{\frac{1}{2}}}}
\newcommand{\fourth}{{\textstyle{\frac{1}{4}}}}

\newcommand{\ox}{{\omega_X}}
\newcommand{\oc}{{\omega_C}}
\newcommand{\olp}{{\omega^0_\text{LP}}}
\newcommand{\oup}{{\omega^0_\text{UP}}}
\newcommand{\wx}{\varpi_X}
\newcommand{\wc}{\varpi_C}
\newcommand{\exciton}{{\psi_X}}
\newcommand{\photon}{{\psi_C}}
\newcommand{\excitonenv}{\phi_X}
\newcommand{\photonenv}{\phi_C}

\begin{document}

\begin{center}
{\bfseries \Large  Lossless Polariton Solitons}
\end{center}

\vspace{0ex}

\begin{center}
{\scshape \large Stavros Komineas${}^*$, Stephen P. Shipman${}^\dagger$\\ and Stephanos Venakides${}^\ddagger$} \\
\vspace{2ex}
{
\itshape
${}^*$Department of Mathematics and Applied Mathematics\\
University of Crete\\
Heraklion, Crete, Greece\\
\vspace{1.2ex}
${}^\dagger$Department of Mathematics\\
Louisiana State University\\
%Lockett Hall 303\\
Baton Rouge, Louisiana \ 70803, USA\\
\vspace{1.2ex}
${}^\ddagger$Department of Mathematics\\
Duke University\\
Durham, North Carolina \ 27708, USA
}
\end{center}

\vspace{3ex}
\centerline{\parbox{0.9\textwidth}{
{\bf Abstract.}
Photons and excitons in a semiconductor microcavity interact to form exciton-polariton condensates.  These are governed by a nonlinear quantum-mechanical system involving exciton and photon wavefunctions.  We calculate all non-traveling harmonic soliton solutions for the one-dimensional lossless system.  There are two frequency bands of bright solitons when the inter-exciton interactions produce an attractive nonlinearity and two frequency bands of dark solitons when the nonlinearity is repulsive.
In addition, there are two frequency bands for which the exciton wavefunction is discontinuous at its symmetry point, where it undergoes a phase jump of $\pi$.
A band of continuous dark solitons merges with a band of discontinuous dark solitons, forming a larger band over which the soliton far-field amplitude varies from $0$ to $\infty$; the discontinuity is initiated when the operating frequency exceeds the free exciton frequency.  The far fields of the solitons in the lowest and highest frequency bands (one discontinuous and one continuous dark) are linearly unstable, whereas the other four bands have linearly stable far fields, including the merged band of dark solitons. 
}}

\vspace{3ex}
\noindent
\begin{mbox}
\end{mbox}
\vspace{3ex}

\hrule
\vspace{3ex}

\section{Introduction} %%%%%%%%%%%%%%%%%%%%%%%%%%%%%%%%%%%

Exciton-polaritons are a quantum-mechanical quasiparticle formed by the coupling of photons with excitons.  An exciton is a dipole generated in a semiconductor when an electron absorbs a photon and jumps from the valence to the conduction band thus leaving a hole in the valence band.  The electron and hole are attracted to each other by an effective electrostatic Coulomb force, resulting in the excitation of an electron-hole pair.  Exciton-polaritons can be trapped in a planar microcavity containing a semiconductor material, that is, they live in a two-dimensional quantum well.

Exciton-polaritons  can form Bose-Einstein condensates (BEC) at relatively high temperatures
\cite{CarusottoCiuti_RMP2013,DengHaug_RMP2010,KavokinBaumberg2007,KeelingMarchetti_SST2007}, sustained by continuous laser pumping of photons. The condensate wavefunctions produce a rich variety of localised quantum states in the micrometer scale: dark solitons \cite{AmoPigeon_Science2011,GrossoNardin_PRB2012,HivetFlayac_NatPhys2012,LarionavoaStolz_OptLett2008,PigeonCarusottoCiuti_PRB2011}, bright solitons \cite{EgorovGorbach_PRL2010,EgorovSkryabinYulin2009,LarionavoaStolz_OptLett2008,SichKriszanovskii_NatPhot2012}, and vortices \cite{GrossoNardin_PRL2011,MarchettiSzymanska}.
Solitons in polaritonic condensates have potential for applications in ultrafast information processing~\cite{AckemannFirthOppo_2009} due to picosecond response times and strong nonlinearities \cite{EgorovSkryabinYulin2009,SichKriszanovskii_NatPhot2012}.
See \cite{SnokeLittlewood2010}, for example, for a tour of polariton condensates.

In a mean-field approximation, the excitons and photons are described by separate wavefunctions $\exciton$ (excitons) and $\photon$ (photons) of spatial coordinates $\,\mathbf{x}=(x_1,x_2)$ and time $t$.  A continuous absorption and emission of photons by atoms in the semiconductor (Rabi oscillation) is represented by a coupling of the two equations.  The kinetic term (Laplacian) is typically neglected for the excitons due to their significantly larger mass. On the other hand, exciton-exciton interaction is significant, so a nonlinear term arises in the equation for the exciton field. The system of equations reads \cite{CarusottoCiuti2004,Deveaud2007,KavokinBaumbergMalpuech2007,MarchettiSzymanska2011,YulinEgorovLederer2008}
\begin{equation}\label{2Dpolaritons}
  i\partial_t
  \renewcommand{\arraystretch}{1.3}
\left(\hspace{-5pt}
  \begin{array}{c}
    \exciton \\ \photon
  \end{array}
\hspace{-5pt}\right)
=
  \renewcommand{\arraystretch}{1.1}
\left(\hspace{-5pt}
  \begin{array}{cc}
    \ox - i\kappa_X + g|\exciton|^2 & \gamma \\
    \gamma & \omega_C - i\kappa_C - \frac{\hbar}{2m_C}\triangle
  \end{array}
\hspace{-5pt}\right)
  \renewcommand{\arraystretch}{1.3}
\left(\hspace{-5pt}
  \begin{array}{c}
    \exciton \\ \photon
  \end{array}
\hspace{-5pt}\right)
+
  \renewcommand{\arraystretch}{1.3}
\left(\hspace{-5pt}
  \begin{array}{c}
    0 \\ F
  \end{array}
\hspace{-5pt}\right).
\end{equation}
The numbers $\omega_{X,C}$ and $\kappa_{X,C}$ are real; $\ox$ is the frequency of a free exciton, $\oc$ is the frequency of the free, zero-momentum photon; $\kappa_X$ and $\kappa_C$ are the attenuation constants of the exciton and photon and account for losses; $m_C$ is the mass associated with the photons.
The Laplacian is denoted by $\triangle=\partial_{x_1}^2+\partial_{x_2}^2$.
The forcing $F$ represents a pumping of photons into the microcavity.  The coupling associated with the Rabi oscillations enters through the frequency parameter $\gamma$, which is half the Rabi frequency.  The nonlinearity $g|\exciton|^2$ is attractive when {$g<0$} and repulsive when~{$g>0$}.

\smallskip

This work addresses the analytic foundations of exciton-polaritons, as a complement to the large body of experimental and numerical work on the subject.
We investigate polariton fields that are lossless ($\kappa_X$ and $\kappa_C$ are zero) and unforced ($F=0$).
In turning off both pumping and losses, which are due to radiation and thermalization, we focus on the synergy of exciton interaction (nonlinearity) and photon dispersion.
We consider fields that depend on only one spatial variable, say $x_1$, and we use the notation $x=x_1$ below.
Under these conditions, we discover and analyze three families solitons that exhaust all harmonic stationary (non-traveling) exciton-polariton solitons.
We use the term ``soliton'' in a broad sense to refer to a field with amplitude that tends to a constant value as $|x|\to\infty$, which is typical in the physics literature.
All solitons can be expressed exactly by quadrature through exact integration of the harmonic polariton system (see~(\ref{darksolution1})--(\ref{darksolution2})).

For each of the three families, the solitons exist on two frequency bands, all six bands being mutually disjoint (Fig.~\ref{fig:bands}).  There is one family of dark solitons for $g>0$, one family of bright solitons for $g<0$, and one family of solitons whose exciton wavefunction exhibits a spatial jump discontinuity.  The discontinuity is made physically possible by the vanishing of the photon field, which brings dispersion to the system, at the point of discontinuity of the exciton field; mathematically, these are distributional solutions of the polariton equations.

The stationary dark solitons in bands 1.1 and 3.2 were introduced in \cite{KomineasShipmanVenakides2015}.
The present work exhausts all soliton solutions and thus proves that these two bands (collectively called band D) are the only dark stationary 1D solitons with stable far-field values.
When $\oc>\ox$ (positive detuning), band 1.1 coincides with the ``lower polariton" band of linear ($g\!=\!0$) homogeneous polaritons .  The upper band 1.2 of dark solitons has the same minimal frequency as the ``upper polariton" band of homogeneous linear solitons.  The dispersion relations for the lower and upper bands of linear homogeneous polaritons are shown in \cite[Fig.~1]{MarchettiSzymanska2011}.

\begin{figure}[t]
\centerline{\scalebox{0.6}{\includegraphics{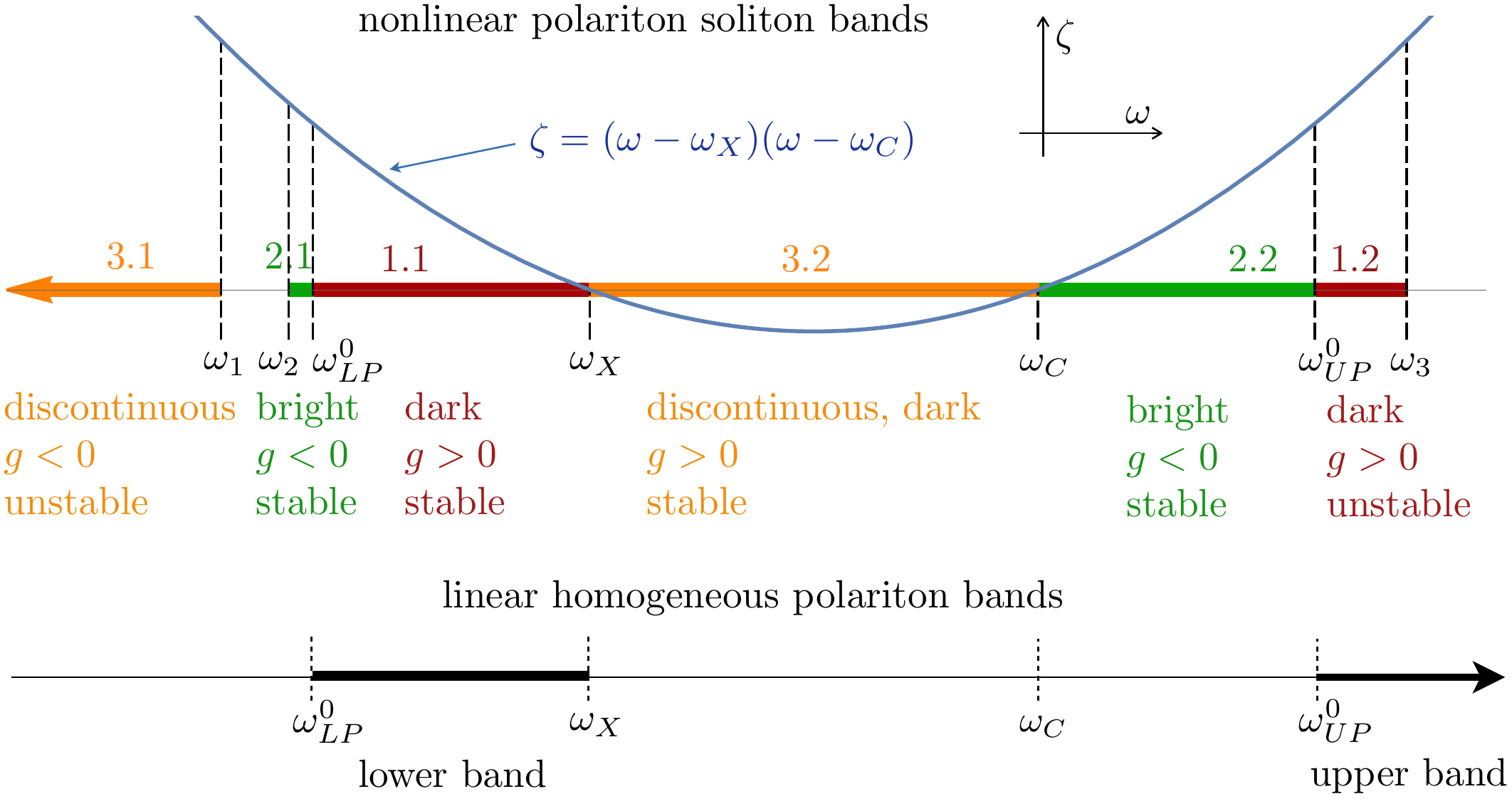}}}
\caption{\small The lossless, unforced, one-dimensional polariton equations admit six frequency bands of stationary soliton-type solutions, whose graphs are shown in Fig.~\ref{fig:field2} and in section~\ref{sec:graphs}.  Positive detuning $\oc>\ox$ is shown here; for negative detuning $\oc<\ox$, band 3.2 is absent (Fig.~\ref{fig:bandsneg}).
Solitons in bands 1.1, 1.2, and 3.2 are all dark and can coexist in a system with $g>0$.  Solitons in bands 2.1, 2.2, and 3.1 can coexist in a system with $g<0$.  The linear stability of the far-field value of the soliton as a constant-amplitude solution of the polariton equations is indicated.
Band 1.1 $= (\olp,\ox)$ coincides with the lower band of homogeneous linear solitons when $\oc>\ox$, and band 1.2 begins at the minimal value of the upper band $(\oup,\infty)$ of homogeneous linear solitons.
The endpoint frequencies of the bands are defined as follows.
Set $p(\omega):=(\omega-\ox)(\omega-\oc)$, with $\ox$ and $\oc$ defined after (\ref{2Dpolaritons}),
$p(\omega_1)=p(\omega_3)=\frac{3}{2}\gamma^2$ and $p(\omega_2)=\frac{9}{8}\gamma^2$ and $p(\olp)=p(\oup)=\gamma^2$.
}
\label{fig:bands}
\end{figure}

\begin{figure}[H]
\centerline{\scalebox{0.53}{\includegraphics{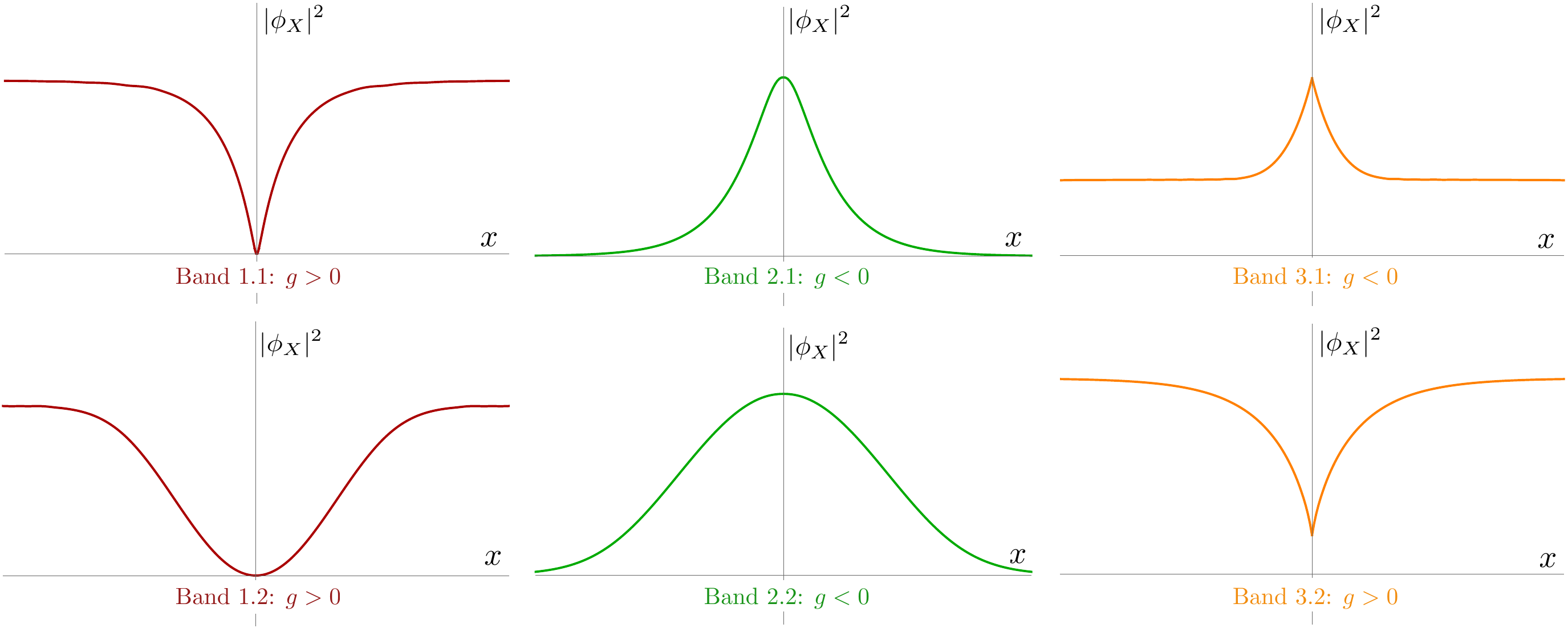}}}
\caption{\small The square modulus $|\exciton(x)|^2$ of the exciton field of the solitons in the bands depicted in Fig.~\ref{fig:bands}.  The exciton fields $\exciton(x)$ of the solitons in bands 3.1 and 3.2 are antisymmetric and experience a discontinuity at the point where $|\exciton(x)|^2$ has a v-shape (see Fig.~\ref{fig:band3}) (note that the soliton of band 1.1 is smooth, although sharp, at its nadir).  The far-field asymptotic value of the solitons in bands 1.1, 1.2, 3.1, and 3.2 is $\zeta_\infty/g$ (defined in Eq.~(\ref{zetainfty})).
Bands 1.1 and 3.2 are unified into one band, as described in section~\ref{subsec:bandD}.}
\label{fig:field2}
\end{figure}

We may interpret the solitons derived in this work as ideal analytic descriptions of localized polariton formations in high-Q microcavities (see, {\itshape e.g.}, \cite{NelsenLiuSteger2013}). It is plausible that, within the finite region of a physical microcavity, the far-field values may be maintained by a small pumping.
%We interpret the solitons derived in this work as ideal analytic descriptions of localized polariton formations in high-Q microcavities.  It is plausible that, within the finite region of a physical microcavity, the far-field values may be maintained by pumping.  Note that high-Q polariton condensates have been reported in \cite{NelsenLiuSteger2013}.
Ref. \cite{TosiChristmannSavvidis_NatPhys2012} reports the creation of polariton condensates at two pump spots, and localized structures can be sustained in the region between the two spots where there is no pumping.
In~\cite{AmoPigeon_Science2011,GrossoNardin_PRB2012}, quasi-one-dimensional polariton structures are observed outside the pump spots.

\section{Lossless harmonic polaritons: reduction to ODE} %%%%%%%%%%%%%%%%%%%%%%%%%%%

Consider a lossless, unforced, one-dimensional polariton field, consisting of a photon wavefunction $\photon(x,t)$ and an exciton wavefunction $\exciton(x,t)$ dynamically coupled through their standard quantum-mechanical equations,
\begin{eqnarray}
  i\partial_t\exciton &=& \left(\omega_X + g|\exciton|^2\right)\exciton + \gamma\photon\,, \label{exciton}\\
  i\partial_t\photon &=& \left( \omega_C - \half \partial_{xx} \right) \photon + \gamma\exciton\,. \label{photon}
\end{eqnarray}
obtained by restricting (\ref{2Dpolaritons}) to one spatial dimension and setting $\kappa_{X,C}=0$ and $F=0$.  The time variable $t$ is normalized to an arbitrary unit of time $T$, frequencies (including $\ox$, $\oc$, $\gamma$, and $g$) are normalized to $1/T$, and the spatial variable $x$ is normalized to $\sqrt{T\hbar/m_C\,}$.  Thus all variables and parameters are non-dimensional.

The polariton equations (\ref{exciton},\ref{photon}) 
admit two quantities that are conserved in time,
\begin{eqnarray}
  N &=& \int \left( |\exciton|^2 + |\photon|^2 \right) dx , \\
  H &=& \int \left( {\textstyle\frac{1}{2}}|\partial_x\photon|^2 + \oc|\photon|^2 + \ox|\exciton|^2 + {\textstyle\frac{g}{2}}|\exciton|^4 + 2\gamma\,\mathrm{Re}(\exciton\photon^*) \right) dx.
\end{eqnarray}

A traveling-wave polariton field with carrier frequency $\omega$, modulated by an envelope has the form
\begin{eqnarray}
  \exciton(x,t) &=& \excitonenv(x-ct) e^{i(kx-\omega t)}, \\
  \photon(x,t) &=& \photonenv(x-ct) e^{i(kx-\omega t)}.
\end{eqnarray}
Under this ansatz, the polariton equations are equivalent to the pair
\begin{eqnarray} \label{traveling_solitons}
  -ic\excitonenv' &=& \left( \ox-\omega + g|\excitonenv|^2 \right)\excitonenv + \gamma\photonenv\,, \\
   i\left( k - c \right)\photonenv' &=& \textstyle\left( \oc-\omega + \frac{k^2}{2} \right)\photonenv - \frac{1}{2}\photonenv'' + \gamma\excitonenv,
\end{eqnarray}
in which the prime denotes the derivative with respect to the argument.

This system of two complex ODEs reduces to a system of two real ODEs when the polariton envelope depends only on the spatial variable (speed of travel $c\!=\!0$) and the polariton  carrier phase  is spatially invariant  (wavenumber $k=0$):
\begin{equation}\label{stationarysolution1}
  \left( \exciton(x,t), \photon(x,t) \right) = \left( \excitonenv(x), \photonenv(x) \right) e^{-i\omega t}.
\end{equation}
Under this assumption, and with the notation
\begin{equation*}
  \wx=\omega-\ox\,,
  \qquad
  \wc=\omega-\oc\,,
\end{equation*}
the pair {of real functions} $(\excitonenv,\photonenv)$ satisfies the equations
\begin{eqnarray}
  \left( g\excitonenv^2 - \wx \right)\excitonenv + \gamma\photonenv &=& 0\,,\label{envelopeODE1a}\\
  -\half \photonenv'' - \wc\photonenv + \gamma\excitonenv &=& 0\,.\label{envelopeODE1b}
\end{eqnarray}
The first equation fully determines  the photon field $\photonenv$ as an odd cubic polynomial function of the exciton field $\excitonenv$, illustrated in Fig.~\ref{fig:cubic}.  Thus the field value pair $(\excitonenv(x),\photonenv(x))$ runs along the graph of the cubic as the spatial variable $x$ varies.  For the solitons in bands 1.1, 1.2, 3.1, and 3.2, the field pair lies on this cubic between the two nonzero equilibrium points $(\excitonenv^\infty,\photonenv^\infty)$ and $(-\excitonenv^\infty,-\photonenv^\infty)$ of the system of equations (\ref{envelopeODE1a},\ref{envelopeODE1b}), where
\begin{equation}\label{eqsoln}
\begin{split}
  \excitonenv^\infty &= \sqrt{\frac{1}{g}\left( \wx - \frac{\gamma^2}{\wc} \right)\,}\,, \\
  \photonenv^\infty &= \frac{\gamma}{\wc} \excitonenv^\infty\,.
\end{split}
\end{equation}
These equilibrium points are indicated by open dots in the graphs in section~\ref{sec:graphs}.
They correspond to homogeneous (spatially constant and time-harmonic) solutions of the polariton equations (\ref{exciton},\ref{photon}),
\begin{equation}\label{homogeneous}
  (\exciton(x,t),\photon(x,t)) \,=\, \pm (\excitonenv^\infty,\photonenv^\infty) e^{-i\omega t}\,.
\end{equation}

\begin{figure} %%%%%%%%%%%%%%%%%%%%%%%%%%%%%%%%%%%%%%%%%%%%
\centerline{\small
\begin{tabular}{ c | c | c |}
   & $\wx<0$ & $\wx>0$ \\ \hline
  $g<0$ 
  & \parbox{0.18\textwidth}{\centerline{\scalebox{0.2}{\includegraphics{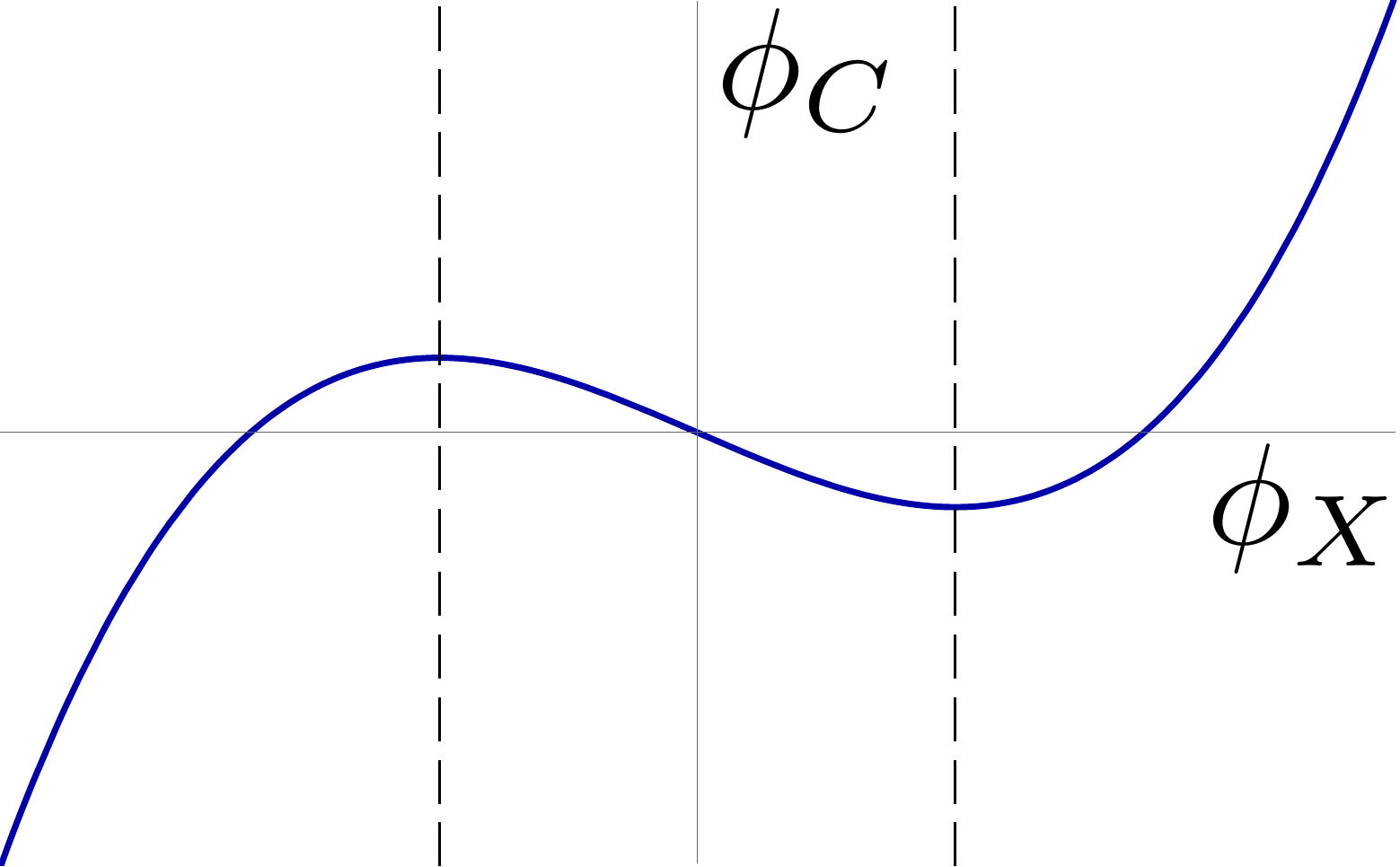}}}}
  & \parbox{0.18\textwidth}{\centerline{\scalebox{0.2}{\includegraphics{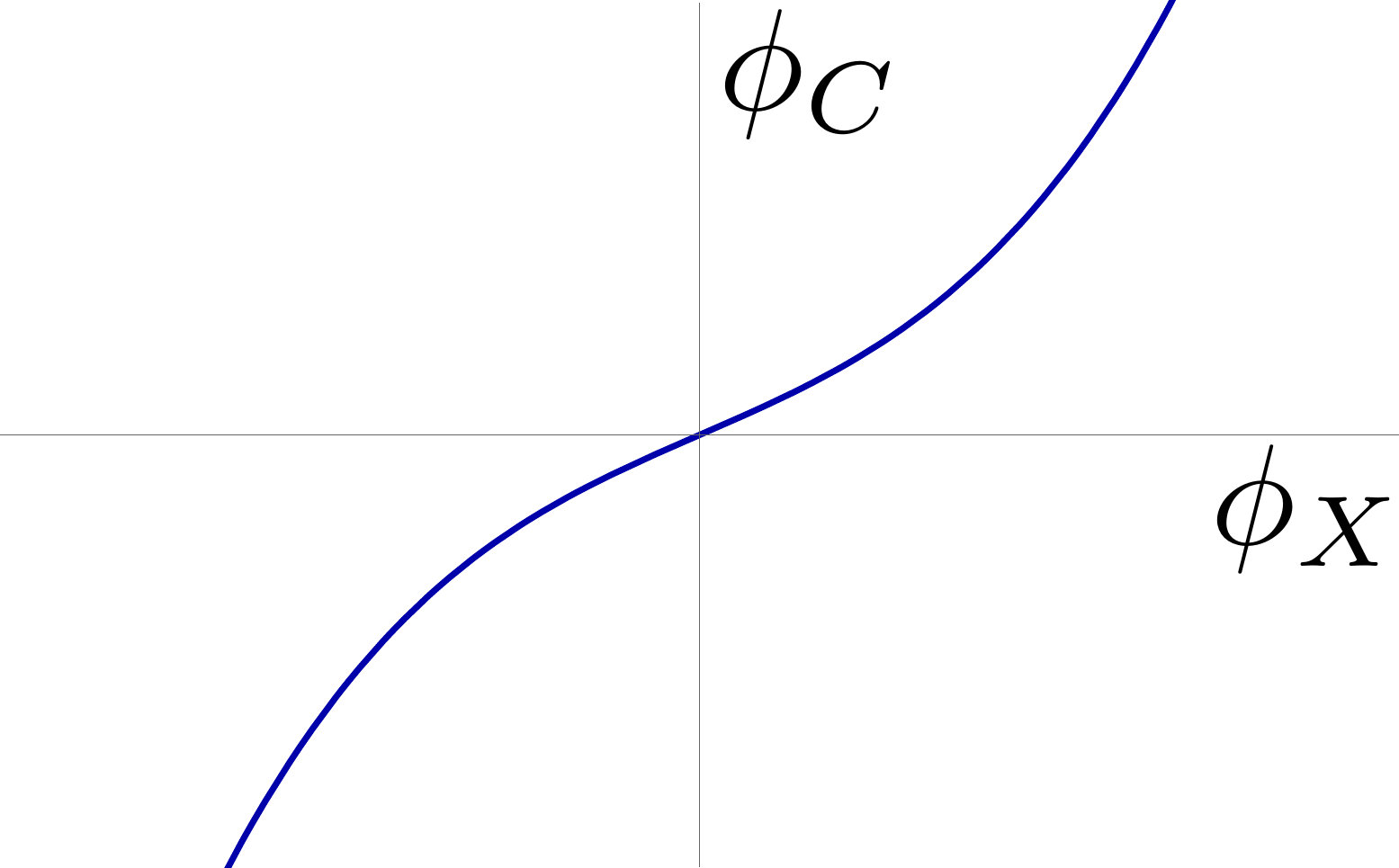}}}}
  \\ \hline
  $g>0$
  & \parbox{0.18\textwidth}{\centerline{\scalebox{0.2}{\includegraphics{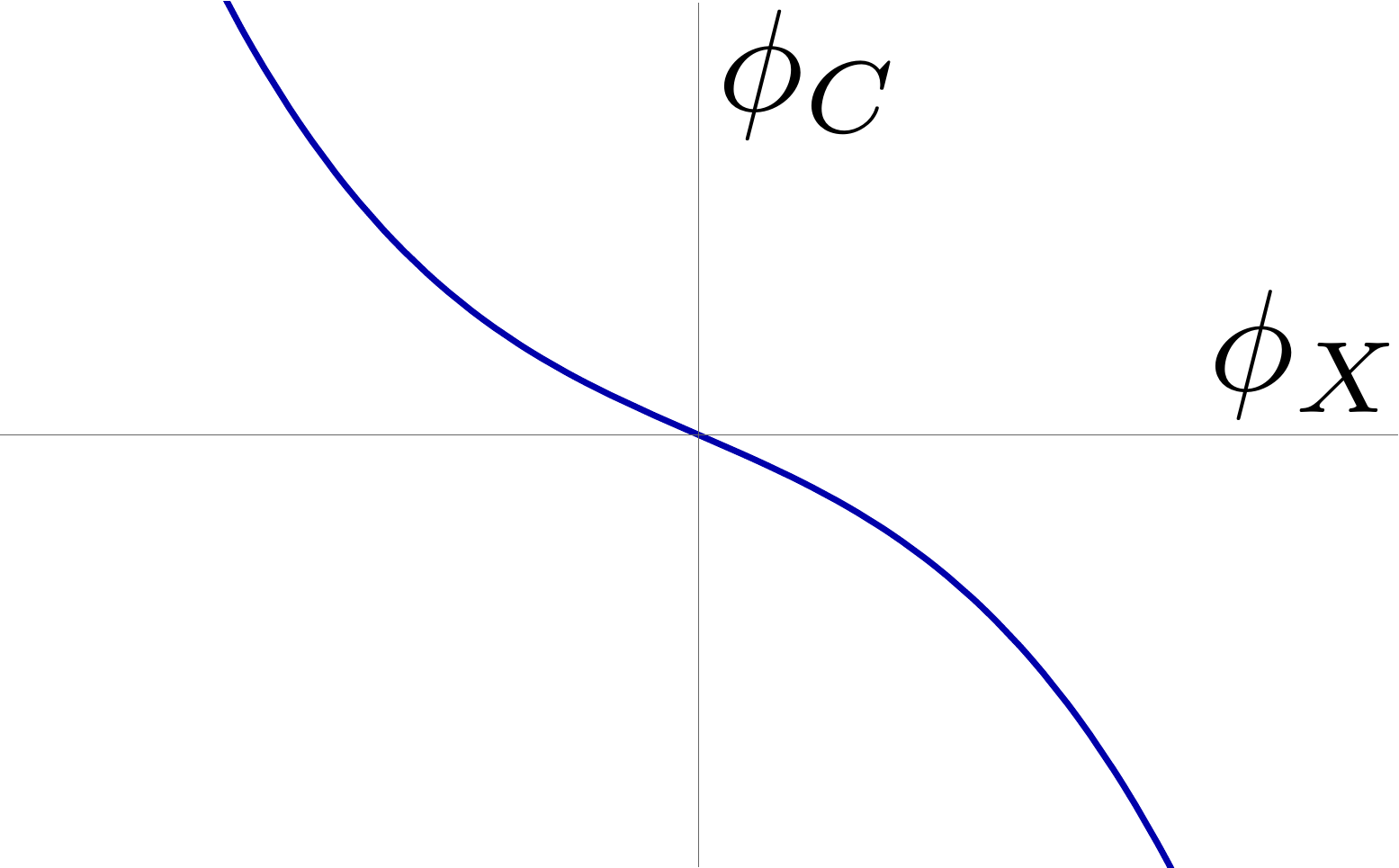}}}}
  & \parbox{0.18\textwidth}{\centerline{\scalebox{0.2}{\includegraphics{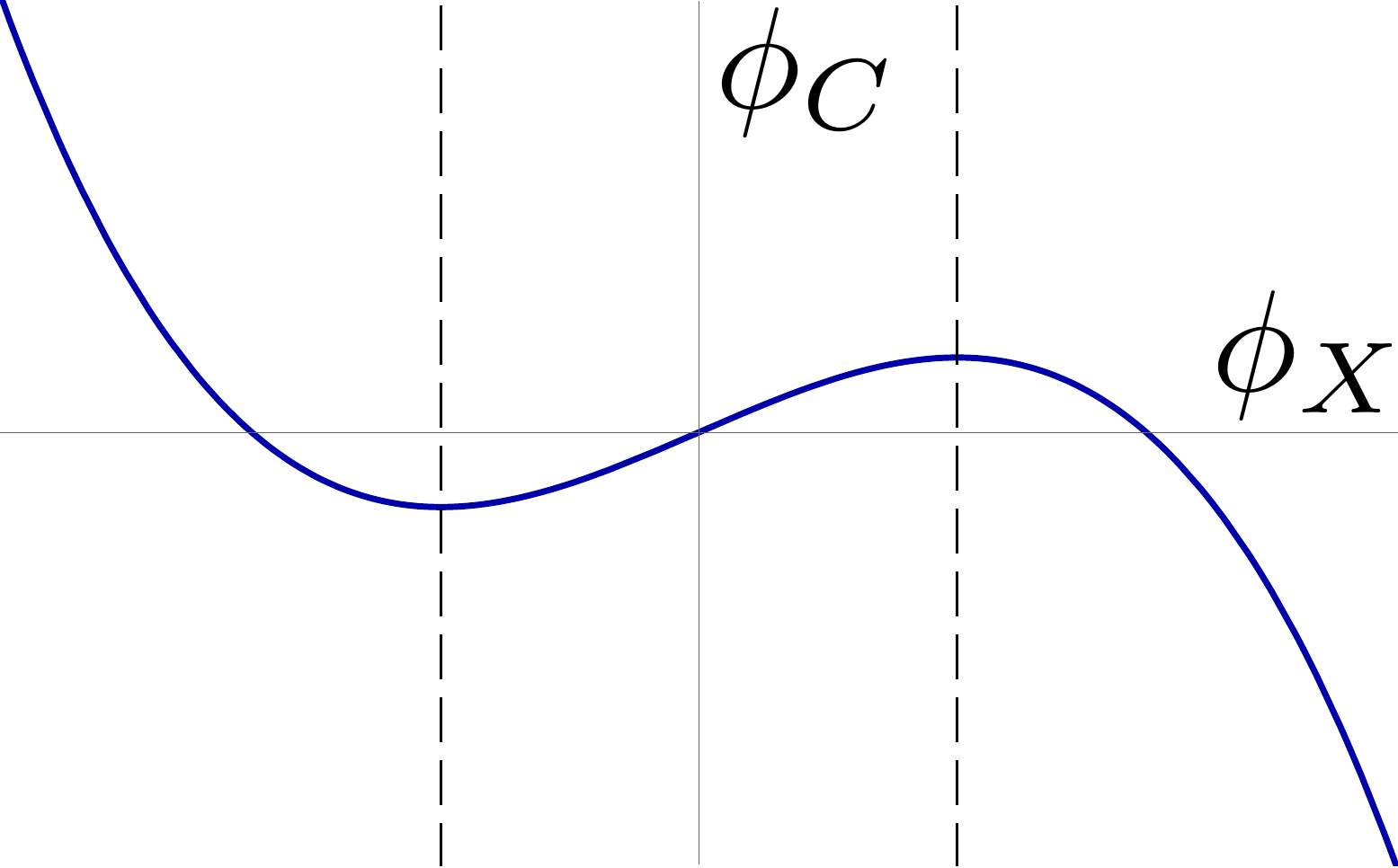}}}}
  \\ \hline  
\end{tabular}
}
\caption{\small The relation (\ref{phirelation}) that gives the photon envelope value $\photonenv$ {\itshape vs.} the exciton envelope value $\excitonenv$ of a harmonic solution of the form $\left( \exciton(x,t), \photon(x,t) \right) = \left( \excitonenv(x), \photonenv(x) \right) e^{-i\omega t}$ of the polariton equations (\ref{exciton},\ref{photon}).  Its shape depends on the signs of $g$ and $\wx=\omega-\ox$.  A soliton solution cannot cross the vertical dotted lines through the critical points.}
\label{fig:cubic}
\end{figure}

In the case that $g\wx<0$, equation (\ref{envelopeODE1a}) defines a monotonic relation between $\photonenv$ and $\excitonenv$ (the off-diagonal graphs in Fig.~\ref{fig:cubic}).  Thus (\ref{envelopeODE1b}) can be written in the form
\begin{equation}\label{hamiltonianODE}
  \photonenv''(x) + U'(\photonenv(x)) = 0,
\end{equation}
which results in the conservation of an energy-type function,
\begin{equation}\label{hamiltonian}
  \half (\photonenv')^2 + U(\photonenv) = K,
\end{equation}
with $K$ being an arbitrary constant.  This equation describes orbits of the system in the $(\photonenv,\photonenv')$-plane (phase plane); homoclinic and heteroclinic orbits correspond to solitons.  In the case that $g\wx>0$, the relation (\ref{envelopeODE1a}) between $\photonenv$ and $\excitonenv$ has three monotonic branches (the diagonal graphs in Fig.~\ref{fig:cubic}) separated by two critical points, where $\photonenv$ achieves a local maximum or minimum as a function of $\excitonenv$.  These two critical points occur when
\begin{equation}
  g\excitonenv^2 = \wx/3 \qquad \text{(critical points)}.
\end{equation}
Within the domain of each branch separately, an equation of the form (\ref{hamiltonian}) holds.

The system (\ref{envelopeODE1a},\ref{envelopeODE1b}) can in fact admit a solution that passes through either critical point, from one branch of the cubic (\ref{envelopeODE1a}) to another.  But such a solution is unique---there is no arbitrary constant of integration analogous to the constant $K$ in (\ref{hamiltonian}).  Moreover, a solution passing through a critical point cannot be a soliton; it is either periodic or becomes unbounded as $|x|\to\infty$.  This is stated in part (c) of Theorem~\ref{thm:zeta}.

One therefore knows that all soliton solutions of (\ref{envelopeODE1a},\ref{envelopeODE1b}) are confined to a single monotonic branch of the cubic (\ref{envelopeODE1a})---the exciton envelope function $\excitonenv(x)$ cannot pass through the values $\sqrt{\wx/(3g)}$.  We will refer to this as the {\em connectivity condition} for polariton solitons.

\smallskip
Analysis of all soliton solutions of (\ref{envelopeODE1a},\ref{envelopeODE1b}), especially with regard to their dependence on $\wx$, $\wc$, and $g$, is complex.  We find that working with the variable
\begin{equation}
 \zeta=g\excitonenv^2
\end{equation}
renders the analysis most transparent.
The system (\ref{envelopeODE1a},\ref{envelopeODE1b}), once integrated, becomes a first-order ODE (Theorem~\ref{thm:zeta}) that expresses $(\zeta')^2$ as a rational function of $\zeta$, in which the nonlinearity parameter $g$ is no longer present.  The forbidden value $\zeta = \wx/3$, occurring at the critical points of the cubic (\ref{envelopeODE1a}), is manifest as a singularity of the equation
\begin{equation}
  (\zeta')^2 = \frac{8}{9} \frac{\zeta Q(\zeta)}{\left(\zeta - \frac{\wx}{3} \right)^2}\,.
\end{equation}
($Q$ is a cubic polynomial defined below.) The sign of $g$ is determined by the sign of $\zeta(x)$, which is constant for any solution.

In terms of $\zeta$, the equilibrium points of (\ref{envelopeODE1a},\ref{envelopeODE1b}) are expressed as $g(\excitonenv^\infty)^2=\zeta_\infty$, where $\zeta_\infty$ is the non-dimensional frequency
\begin{equation}\label{zetainfty}
  \zeta_\infty \,=\, \zeta_\infty(\omega) \,:=\, \wx - \frac{\gamma^2}{\wc}
  \,=\, \omega-\ox - \frac{\gamma^2}{\omega-\oc}  \,.
\end{equation}
\smallskip

\begin{theorem}\label{thm:zeta}
(a) The pair of equations (\ref{envelopeODE1a},\ref{envelopeODE1b}) implies the pair
\begin{eqnarray}
    && (3\zeta - \wx)^2\,\zeta'^2 \,=\, 8\, \zeta\,Q(\zeta)\,, \label{zetaODE} \\
    &&  \photonenv \,=\, \textstyle\frac{1}{\gamma} \left( \wx - g\excitonenv^2 \right)\excitonenv\,, \label{phirelation}
\end{eqnarray}
in which $\zeta(x)=g\excitonenv^2(x)$ and $Q(\zeta)$ is a cubic polynomial in $\zeta$,
\begin{equation}\label{Q_polynomial}
  Q(\zeta) = Q(\zeta;\omega) = -\wc\left( \zeta^3 - \textstyle\frac{1}{2}(3\zeta_\infty+\wx)\,\zeta^2 + \zeta_\infty\wx\,\zeta + K \right)
\end{equation}
and $K$ is a constant of integration.
Conversely, whenever $\excitonenv'(x)\photonenv'(x)\not=0$, the pair (\ref{zetaODE},\ref{phirelation}) implies the pair (\ref{envelopeODE1a},\ref{envelopeODE1b}).

(b)  Whenever $g$ and $\wx$ have the same sign, the local extrema of
\,$\photonenv = \gamma^{-1} \left( \wx - g\excitonenv^2 \right)\excitonenv$\, occur when $\,\zeta = g\excitonenv^2=\wx/3$.
The roots of $Q'$ are $\zeta_\infty$ and $\wx/3$.
Thus, whenever $\wx/3$ is a root of $Q$, the factor $(3\zeta - \wx)^2$ appears on both sides of (\ref{zetaODE}).

(c) If a solution $(\excitonenv(x),\photonenv(x))$ of (\ref{envelopeODE1a},\ref{envelopeODE1b}) passes through a critical point of the cubic (\ref{envelopeODE1a}), then either the solution is periodic or $|\photonenv(x)|$ and $|\excitonenv(x)|$ tend to $\infty$ as $|x|\to\infty$.

\end{theorem}

\begin{proof}
To prove these statements, the structure of equations (\ref{envelopeODE1a},\ref{envelopeODE1b}) is illuminated by writing them as
\begin{eqnarray}
  f_1(\excitonenv) + \gamma\photonenv &=& 0\,,\label{envelopeODE2a}\\
  f_2(\photonenv) + \gamma\excitonenv &=& \half\photonenv'' \,,\label{envelopeODE2b}
\end{eqnarray}
in which $f_1$ and $f_2$ are odd polynomials.

By multiplying the first equation by $\excitonenv'$ and the second by $\photonenv'$, adding, and then taking antiderivatives, one obtains
\begin{equation}\label{integratedODE1}
  K' + \tilde f_1(\excitonenv) + \tilde f_2(\photonenv) + \gamma\excitonenv\photonenv \,=\, \fourth(\photonenv')^2\,,
\end{equation}
in which the even polynomials $\tilde f_{1,2}$ are primitives of $f_{1,2}$ and $K'$ is an arbitrary constant.
Equation (\ref{envelopeODE2a}) expresses $\photonenv$ as an odd polynomial function of $\excitonenv$, so the left-hand side of (\ref{integratedODE1}) is an even polynomial function of $\excitonenv$, say $P(\excitonenv)$.  Thus (\ref{envelopeODE2a},\ref{envelopeODE2b}) is equivalent to the validity of the pair
\begin{eqnarray}
  P(\excitonenv) &=& \fourth \left( \photonenv' \right)^2 \,,\label{envelopeODE3a}\\
  f_1(\excitonenv) + \gamma\photonenv &=& 0 \,,\label{envelopeODE3b}  
\end{eqnarray}
for some constant $K'$ in the definition of $P$, which is computed to be
\begin{equation}\label{P}
  P(\phi) = K' - \frac{\wc g^2}{2\gamma^2} \phi^6
     + g \left( \frac{\wx\wc}{\gamma^2} - \frac{3}{4} \right) \phi^4
     + \frac{\wx}{2} \left( 1 - \frac{\wx\wc}{\gamma^2} \right) \phi^2\,.
\end{equation}

Equation (\ref{envelopeODE3a}) can be written equivalently in terms of $\excitonenv$ alone by differentiating (\ref{envelopeODE3b}) with respect to $x$ and substituting the resulting expression for $\photonenv'$ into the right-hand side of (\ref{envelopeODE3a}),
\begin{equation}\label{integratedODE2}
  4\gamma^2\, P(\excitonenv) \,=\, f_1'(\excitonenv)^2\, (\excitonenv')^2\,.
\end{equation}

Since $P$ is even, it is a cubic polynomial function of $\zeta = g\excitonenv^2$, and a calculation converts (\ref{integratedODE2}) into the differential equation stated in the theorem,
\begin{equation}\label{zetaODE2}
  {(3\zeta - \wx)^2\zeta'^2 \,=\, 8\, \zeta\,Q(\zeta)\,,}
\end{equation}
in which the cubic polynomial $Q$ is related to $P$ through\, $2g\gamma^2P(\excitonenv)=Q(\zeta)$\, and is given~by
\begin{equation}
  Q(\zeta) 
   \,=\, 2g\gamma^2K' + \zeta\wx(\gamma^2-\wx\wc) - \zeta^2\left( {\textstyle\frac{3}{2}}\gamma^2 - 2\wx\wc \right) - \zeta^3\wc\, .
\end{equation}
Notice that $g$ has disappeared from the differential equation (\ref{zetaODE}) except where it is multiplied by the constant of integration in $Q$.

We must now show that $Q$ has the form given in part (a) of the Theorem and prove part~(b).  For part (b), differentiate (\ref{envelopeODE3a}) with respect to $x$ to obtain
\begin{equation}\label{P2}
  P'(\excitonenv)\excitonenv' \,=\, \half\photonenv'' \photonenv'\,.
\end{equation}
Then, by using (\ref{envelopeODE2b}) and the $x$-derivative of (\ref{envelopeODE3b}), equation (\ref{P2}) is rewritten as
\begin{equation}
 P'(\excitonenv) \,=\,
      -f_1'(\excitonenv)\left( \excitonenv + \gamma^{-1} f_2(\photonenv) \right)
\end{equation}
when $\excitonenv'\not=0$.  But both sides of this equation are polynomials in $\excitonenv$ (recall $\photonenv$ is a polynomial function of $\excitonenv$), so this leads to the observation that {any root of the polynomial $f_1'$ is also a root of~$P'$}.
But $f_1'(\excitonenv) = 3g\excitonenv^2 - \wx$, so 
\begin{equation}\label{Pprime}
  P'(\excitonenv) \,=\, 0
  \qquad
  \text{when}
  \quad
  {g\excitonenv^2=\wx/3}\,.
\end{equation}
From the relation \,$2g\gamma^2P(\excitonenv)=Q(\zeta)$,\, one finds that
\begin{equation}
  Q'(\wx/3)=0 \qquad \text{(if $\wx\not=0$)},
\end{equation}
which proves part (b).

The other root of $Q'$ is found to be
\,$\zeta_\infty = \wx - \frac{\gamma^2}{\wc}$,\, and one can write $Q$ as stated in the theorem (\ref{Q_polynomial}), with
\begin{equation}
  K = -\frac{2g\gamma^2\,K'}{\wc}\,.
\end{equation}

To prove part (c), suppose first that a classical solution of (\ref{envelopeODE1a},\ref{envelopeODE1b}) satisfies $\zeta(0)=g\excitonenv(0)^2 = \wx/3$, so that the first factor of the left-hand side of (\ref{zetaODE}) vanishes at $x=0$.  Thus $Q(\wx/3)=0$, and by part (b), $\wx/3$ is a double root of $Q$, and (\ref{zetaODE}) reduces to
\begin{equation}
  (\zeta')^2 = -\frac{8\wc}{9} \zeta(\zeta-\zeta_0),
\end{equation}
in which
\begin{equation*}
  \zeta_0 = \frac{4\wx}{3} - \frac{3\gamma^2}{2\wc}\,.
\end{equation*}
By standard phase-plane analysis, the solutions of this equation are either constant, periodically oscillating between $\zeta=0$ and $\zeta=\zeta_0$, or tending to $\infty$ as $|x|\to\infty$.

Next, suppose that $(\wx/3)Q(\wx/3)>0$ and that $\zeta(x)$ is continuous with $\zeta(0)=\wx/3$.  We will show that this case is ruled out.
In what follows, the functions $h_i(y)$ are analytic at $y=0$ with $h_i(0)>0$.
Because the numerator in the differential equation
\begin{equation}
  \zeta' \,=\, \pm\frac{\sqrt{8\zeta Q(\zeta)\,}}{3\zeta-\wx}
\end{equation}
is positive at $\zeta=\wx/3$, one obtains, for $x$ near $0$, either
\begin{equation}
 \zeta(x) = \frac{\wx}{3} + |x|^{1/2} h_1(|x|^{1/2})
\end{equation}
or
\begin{equation}
 \zeta(x) = \frac{\wx}{3} - |x|^{1/2} h_1(-|x|^{1/2})\,,
\end{equation}
in which one sign is chosen for $x>0$ and one sign is chosen for $x<0$.
Denote
\begin{equation}
  \tilde\excitonenv(x) = \excitonenv(x) - \excitonenv(0)\,,
  \qquad
  \tilde\photonenv(x) = \photonenv(x) - \photonenv(0)\,,
\end{equation}
with $\excitonenv(0)=\pm\sqrt{\wx/(3g)\,}$.
From $\wx\not=0$ and the equation $\excitonenv(x) = \sqrt{\zeta(x)/g\,}$ (the sign of $g$ is chosen to make the square root real), one obtains
\begin{equation}
  \tilde\excitonenv(x) \;=\; |x|^{1/2} h_2(|x|^{1/2})
  \quad\text{or}\quad - |x|^{1/2} h_2(-|x|^{1/2})\,.
\end{equation}
Since $(\excitonenv(0),\photonenv(0))$ is a critical point of the cubic $\photonenv$ {\itshape vs.} $\excitonenv$ relation, one has
\begin{equation}
  \tilde \photonenv(x) = \pm \tilde\excitonenv(x)^2 h_3(\tilde\excitonenv(x)) = \pm |x| h_4(\pm|x|^{1/2})\,.
\end{equation}
It follows that $\photonenv''(x)$ has a leading singular part equal to a nonzero multiple of the delta function $\delta(x)$, and this is inconsistent with equation (\ref{envelopeODE1b}).
\end{proof}

\section{Polariton solitons}\label{sec:solitons} %%%%%%%%%%%%%%%%%%%%%%%%%%%%%%%

A solution of the ODE (\ref{zetaODE}) corresponds to a stationary, or non-traveling, time-harmonic solution of the polariton system (\ref{exciton},\ref{photon}).  Our interest is in soliton solutions, for which the spatial envelope has a limiting value as $|x|\to\infty$.  All solitons and their frequency bands are described in Theorem~\ref{thm:solitons} below, and proved in section~\ref{sec:proofsolitons}.  The system admits bands of dark and bright solitons.

In addition, we find solitons for which the exciton field is discontinuous at its point of symmetry where the photon field vanishes.  {These are distributional solutions of the soliton equations; this is proved in section~\ref{sec:proofsolitons} (p.\,\pageref{discontinuoussolitons}).}  The physical origin of the discontinuity is that the vanishing of the photon field at a point in space turns off the interaction between neighboring excitons because this interaction is mediated only by the coupling of the exciton field to the dispersive photon field.
When $\ox<\oc$, a band of continuous dark solitons and a band of discontinuous dark solitons can be unified into a single band of dark solitons, as described in section~\ref{subsec:bandD}.

\subsection{Main theorem: description of all solitons}

The reduction of the polariton system (\ref{exciton},\ref{photon}) to an ODE (\ref{zetaODE}) under the assumption of harmonic solutions allows a complete derivation of all stationary soliton solutions of (\ref{exciton},\ref{photon}).  By a stationary soliton solution, we mean a harmonic solution $(\exciton(x,t),\photon(x,t))$ = $(\excitonenv(x),\photonenv(x))e^{-i\omega t}$ for which the envelopes $(\exciton(x),\photon(x))$ tend to far-field values as $x\to\pm\infty$.  The form of the ODE, $\zeta'^2=f(\zeta)$ guarantees that $|\zeta(x)|$ and therefore also $|\excitonenv(x)|$ exhibit a single maximum at the soliton peak or minimum at the soliton nadir.

%%%%%%%%%%%%%%%%%%%%%%%%%   THEOREM   %%%%%%%%%%%%%%%%%%%%%%%%%
\begin{theorem}\label{thm:solitons}
There are bright, dark, and discontinuous stationary soliton solutions of the lossless, unforced polariton equations (\ref{exciton},\ref{photon}) for frequencies within certain bands that depend on $\ox$, $\oc$, and $\gamma$.
The three soliton classes described below exhaust all solutions of the form 
\begin{equation}\label{stationarysolution2}
  \left( \exciton(x,t), \photon(x,t) \right) \,=\, \left( \excitonenv(x), \photonenv(x) \right) e^{-i\omega t}
\end{equation}
for which $\excitonenv(x)$ and $\photonenv(x)$ have limits (far-field values) as $x\to\pm\infty$.
\begin{enumerate}
  \item {\bf Dark solitons.} (Red bands in Figs.~\ref{fig:band1} and \ref{fig:band123}) Equations (\ref{exciton},\ref{photon}), for $g>0$, admit solutions of the form (\ref{stationarysolution2}) for which $\excitonenv(x)$ and $\photonenv(x)$ are antisymmetric, monotonic, and bounded.  The frequency bands for which these solutions exist are given by
\begin{equation}\label{bands1}
  \renewcommand{\arraystretch}{1.0}
\left.
  \begin{array}{ll}
    0 < \wx\wc < \gamma^2 & \text{if\hspace{0.8em} $\omega<\min\{\omega_C,\omega_X\}$}\,,
    \quad \text{(band 1.1)}\\
    \vspace{-2ex}\\
    \gamma^2 < \wx\wc < \frac{3}{2}\gamma^2 & \text{if\hspace{0.8em} $\omega>\max\{\omega_C,\omega_X\}$}\,.
        \quad \text{(band 1.2)}
  \end{array}
\right.
\end{equation}
The far-field (suprimal) value of $|\excitonenv(x)|$ is
\begin{equation}\label{farfield1}
  \lim_{|x|\to\infty} |\excitonenv(x)| \;=\; \frac{1}{\sqrt{g}} \left| \wx - \frac{\gamma^2}{\wc} \right|^{1/2}\,.
\end{equation}
  \item {\bf Bright solitons.} (Green bands in Figs.~\ref{fig:band2} and \ref{fig:band123}) Equations (\ref{exciton},\ref{photon}), for $g<0$, admit solutions of the form (\ref{stationarysolution2}) for which $\excitonenv(x)$ and $\photonenv(x)$ are symmetric and bounded and have a unique local maximum or minimum.  The frequency bands for which these solutions exist are given by
\begin{equation}\label{bands2}
  \renewcommand{\arraystretch}{1.0}
\left.
  \begin{array}{ll}
    \gamma^2 < \wx\wc < \frac{9}{8}\gamma^2 & \text{if\hspace{0.8em} $\omega<\min\{\omega_C,\omega_X\}$}\,,
        \quad \text{(band 2.1)}\\
    \vspace{-2ex}\\
    0 < \wx\wc < \gamma^2 & \text{if\hspace{0.8em} $\omega>\max\{\omega_C,\omega_X\}$}\,.
        \quad \text{(band 2.2)}\\
  \end{array}
\right.
\end{equation}
These solitons vanish at the far field ($|x|\to\infty$), and $|g|\excitonenv(x)^2$ attains a maximal value~of
\begin{equation}\label{brightpeak}
  \max_{-\infty<x<\infty} |g|\excitonenv(x)^2 \,=\,
  - \wx +
  \frac{\gamma^2}{\wc}
  \left[
  \frac{3}{4} + \sqrt{\frac{9}{16}-\frac{\wx\wc}{2\gamma^2}}\,
  \right].
\end{equation}
  \item {\bf Discontinuous solitons.} (Orange bands in Figs.~\ref{fig:band3} and \ref{fig:band123}) Equations (\ref{exciton},\ref{photon}) admit antisymmetric bounded solutions of the form (\ref{stationarysolution2}) for which $\excitonenv(x)$ is discontinuous at $x=0$ but $\photonenv(x)$ is continuous.  They satisfy the polariton equations in the distributional sense, and away from the point of discontinuity they satisfy the equations classically.  These solutions exist in the following two frequency bands.
  \begin{enumerate}
    \item For $g<0$, and all frequencies satisfying
    \begin{equation*}
      \omega<\min\{\omega_C,\omega_X\} \quad \text{and} \quad \wx\wc > \textstyle \frac{3}{2}\gamma^2 \,,
          \quad \text{(band 3.1)}
    \end{equation*}
    there is a soliton such that $|\excitonenv(x)|$ decreases monotonically
    \begin{equation*}
      \text{from} \quad
      \left| \frac{\wx}{g} \right|^{1/2}
      \quad \text{down to} \quad
      \frac{1}{\sqrt{|g|}} \left| \wx - \frac{\gamma^2}{\wc} \right|^{1/2}
    \end{equation*}
as $x$ runs from the location of the peak of $|\excitonenv(x)|$ to $\infty$.
Where $|\excitonenv(x)|$ experiences its peak, $|\photonenv(x)|$ experiences its nadir.
    \item For $g>0$ and all frequencies satisfying
    \begin{equation*}
      \omega_X < \omega < \omega_C \,,
      \quad \text{(band 3.2)}
    \end{equation*}
    there is a dark soliton such that $|\excitonenv(x)|$ increases monotonically
    \begin{equation*}
      \text{from} \quad
      \left| \frac{\wx}{g} \right|^{1/2}
      \quad \text{up to} \quad
      \frac{1}{\sqrt{|g|}} \left| \wx - \frac{\gamma^2}{\wc} \right|^{1/2}
    \end{equation*}
    as $x$ runs from the location of the nadir of $|\excitonenv(x)|$ to $\infty$.
    In the negative detuning case, $\oc<\ox$, this band is absent.
  \end{enumerate}
\end{enumerate}
\end{theorem}
%%%%%%%%%%%%%%%%%%%%%%   END THEOREM   %%%%%%%%%%%%%%%%%%%%%%%%

At the far field, the dark solitons and the discontinuous solitons tend to homogeneous solutions (\ref{homogeneous}) of the polariton equations, that is,
\begin{equation}
  \lim_{|x|\to\infty} g\excitonenv(x)^2=\zeta_\infty(\omega),
\end{equation}
with $\zeta_\infty$ given by (\ref{zetainfty}).  According to Theorem~\ref{thm:zeta}(b), $\zeta_\infty(\omega)$ is one of the stationary points of $Q(\zeta;\omega)$ ({\itshape i.e.}, $\partial Q(\zeta;\omega)/\partial\zeta=0$).

The peak values of the bright soliton amplitudes are expressed through one of the roots of $Q(\zeta)$ when $K=0$ (as will be demonstrated in section~\ref{sec:proofsolitons}), namely
\begin{equation}\label{zetazero}
  \textstyle
  \zeta_0(\omega) \,:=\, \wx - \frac{\gamma^2}{\wc}
                 \left[ \frac{3}{4} + \sqrt{\frac{9}{16} - \frac{\wx\wc}{2\gamma^2}\,}\, \right].
\end{equation}
Since $g<0$, $\zeta(x)$ is negative and, according to (\ref{brightpeak}), attains a minimal value of $\zeta_0$.

Band 1.1 of dark solitons coincides with the lower band of linear homogeneous polaritons $(\olp,\ox)$, and band 1.2 starts at the same minimal frequency $\oup$ as that of the upper band of linear homogeneous polaritons \cite[Fig.~1]{MarchettiSzymanska2011}.  The frequencies $\olp$ and $\oup$ are the roots of the quadratic $(\omega-\ox)(\omega-\oc)-\gamma^2$, with $\olp<\oup$.

Band 3.1 of discontinuous solitons for $g<0$ is unusual in that the exciton field exhibits a peak whereas the photon field exhibits a dip at the symmetry point (Fig.~\ref{fig:band3}).

Exact expressions by quadrature can be given for all solitons.  For the dark solitons of bands 1.1 and 1.2, one has
\begin{equation}\label{darksolution1}
  x = \pm \int_0^{\zeta(x)} \frac{3z-\wx}{\sqrt{8zQ(z)}\,}\, dz
  \qquad
  (0\leq\zeta(x)\leq\zeta_\infty).
\end{equation}
The exciton and photon fields are then obtained by
\begin{equation}\label{darksolution2}
\begin{split}
  \excitonenv(x) &\,=\, \pm\, \mathrm{sgn}(x) \sqrt{\frac{\zeta(x)}{g}}\,, \\
  \photonenv(x) &\,=\, \textstyle\frac{1}{\gamma} \left( \wx - g\excitonenv(x)^2 \right)\excitonenv(x)\,.
\end{split}
\end{equation}
In band 3.2, expression (\ref{darksolution1}) is modified by replacing the lower limit of integration by $\wx$ and allowing $\wx\leq\zeta(x)\leq\zeta_\infty$.  Similar expressions apply for the other soliton bands.

\subsection{Graphical depiction of solitons}\label{sec:graphs}

The figures in this section depict the soliton solutions of the form (\ref{stationarysolution2}) for the polariton equations (\ref{exciton},\ref{photon}).  The three types of solitons announced in Theorem~\ref{thm:solitons} are depicted in three separate figures below.

\smallskip
In Figures~\ref{fig:band1}, \ref{fig:band2}, and \ref{fig:band3}, assume that the symmetry point of each soliton (peak or nadir of the amplitude) is at $x=0$.  In each figure:

\begin{spslist}
  \item  The {leftmost diagram} shows two frequency bands of solitons on the $\omega$-axis of the $\omega$$\zeta$-plane.  At a chosen frequency in each band, an arrow spans the range of $\zeta$-values of a soliton, pointing toward the far-field value $\lim_{|x|\to\infty}\zeta(x)$.  The tail of the arrow, indicated by a solid dot, has its ordinate at $\zeta(0)$.  The point of the arrow, indicated by an open circle, has its ordinate at the far-field value of $\zeta(x)$.
  \item  The sign of $g$ coincides with the sign of $\zeta=g\excitonenv^2$.
  \item  The middle and rightmost diagrams depict the exciton and photon envelopes $\excitonenv(x)$ and $\photonenv(x)$ for each frequency corresponding to the arrows in the leftmost diagram.
  \item  The upper graphs depict the trajectory of the point $(\excitonenv(x),\photonenv(x))$ along the cubic relation (\ref{phirelation}) as $x$ traverses the real line.  The solid dots mark the central point $(\excitonenv(0),\photonenv(0))$, and the open circles mark the far-field values $\lim_{x\to\pm\infty}(\excitonenv(x),\photonenv(x))$.
  \item  The lower graphs depict the exciton and photon envelopes {\itshape vs.} the spatial variable $x$.  When one passes from $\zeta=g\excitonenv^2$ to $\excitonenv$, the extraction of square roots results in two solitons, which are minuses of each other.  One choice of square root is shown in the graphs.
  \item  In each figure, $\gamma=1$ and $\oc-\ox=1$.  The inequality $\oc>\ox$ is referred to as ``positive detuning".  The ``negative detuning" case $\oc<\ox$ is depicted in Fig.~\ref{fig:band123} (right) and in Fig.~\ref{fig:bandsneg}.
\end{spslist}

\begin{figure}[H]
\centerline{\scalebox{0.59}{\includegraphics{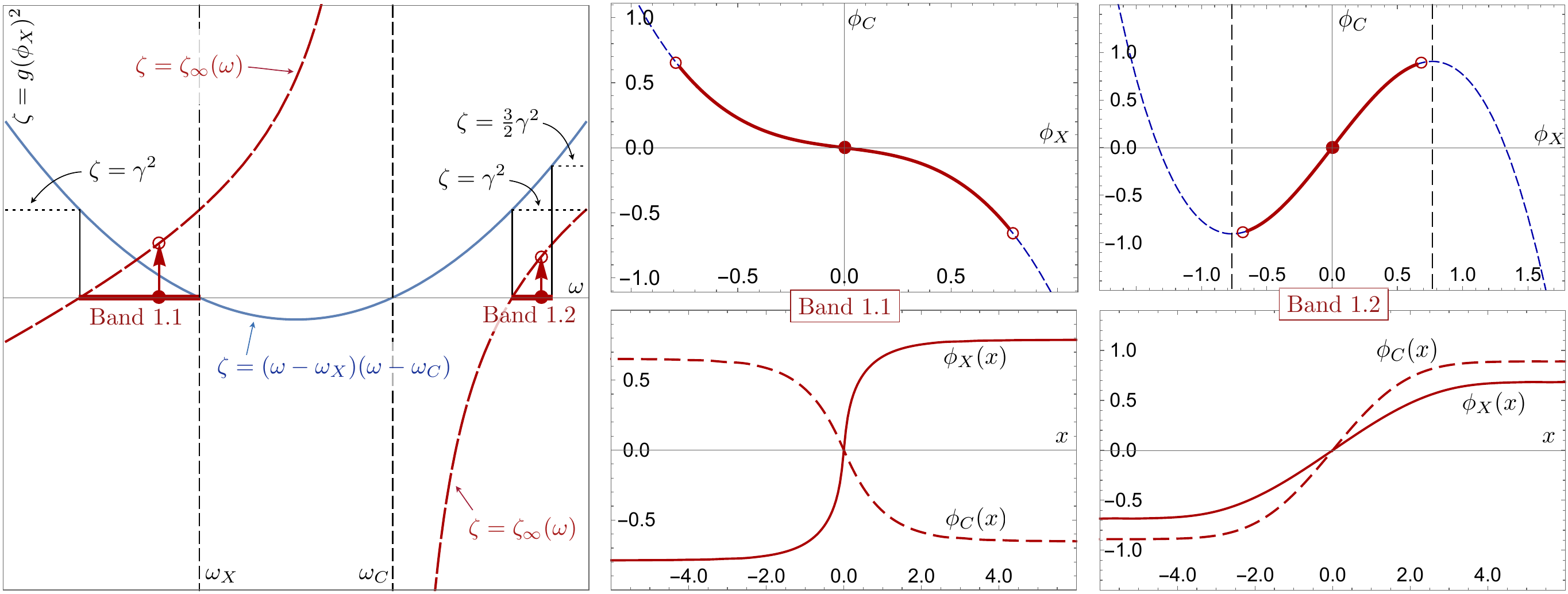}}}
\caption{\small {\bfseries Dark solitons.} Anti-symmetric dark solitons for nonlinearity coefficient $g>0$.
(See the bullet points above in section~\ref{sec:graphs} for a general explanation.) The far-field value $\zeta_\infty$ of $\zeta\!=\!g\excitonenv^2$, indicated by the open dots in the leftmost diagram and given by (\ref{zetainfty}), is equal to a double root of the cubic $Q(\zeta)=Q(\zeta;\omega)$ (see (\ref{Q_polynomial})) created by the appropriate choice of constant~$K=K(\omega)$.  In the upper graphs (middle and right), the pair $(\excitonenv(x),\photonenv(x))$ travels from one open circle to the other as $x$ travels from $-\infty$ to~$\infty$.}
\label{fig:band1}
\end{figure}

\begin{figure}[H]
\centerline{\scalebox{0.59}{\includegraphics{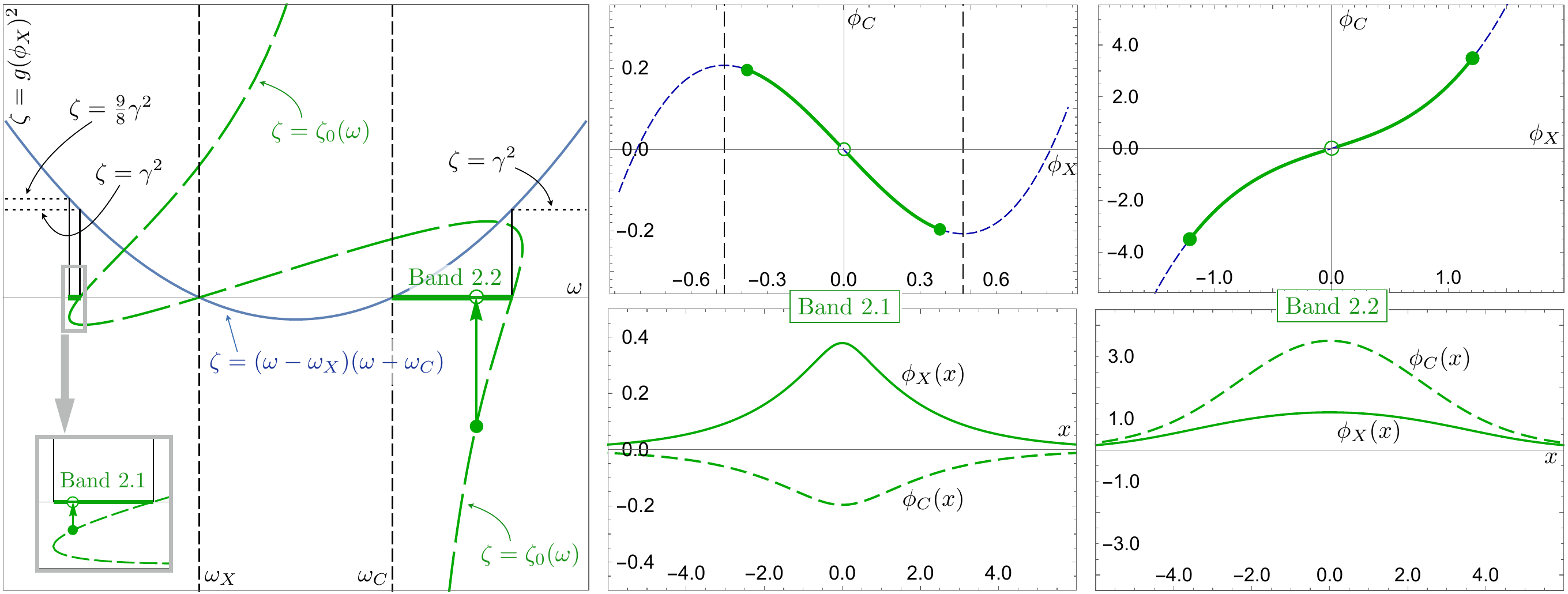}}}
\caption{\small {\bfseries Bright solitons.} Symmetric bright solitons for nonlinearity coefficient $g<0$. (See the bullet points above in section~\ref{sec:graphs} for a general explanation.)  The minimal value $\zeta_0$ of $\zeta(x)$, indicated by the solid dot in the leftmost diagram and given by (\ref{zetazero}), is at a simple root of $Q(\zeta)=Q(\zeta;\omega)$ when $K=0$ so that $\zeta Q(\zeta)$ has a double root at $\zeta=0$.  In the upper graphs (middle and right), the pair $(\excitonenv(x),\photonenv(x))$ travels from the open circle to one of the solid dots and back as $x$ travels from $-\infty$ to~$\infty$.}
\label{fig:band2}
\end{figure}

\begin{figure}[H]
\centerline{\scalebox{0.59}{\includegraphics{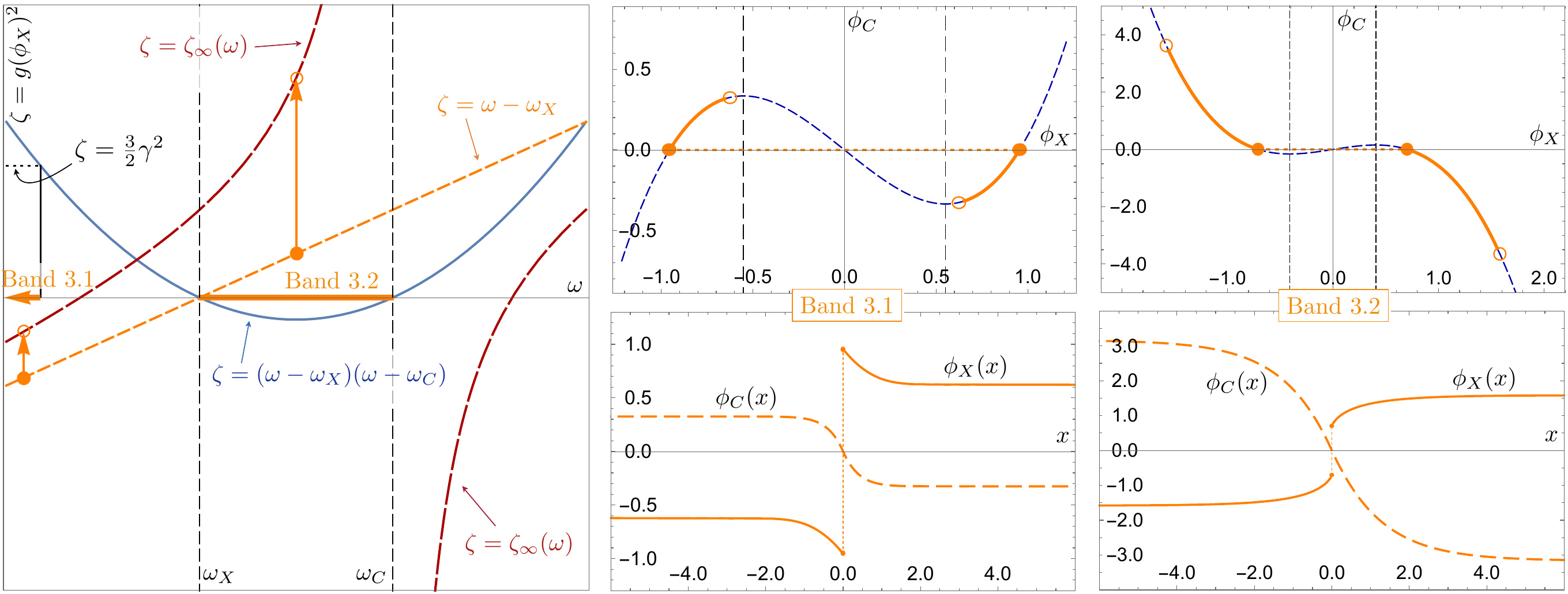}}}
\caption{\small {\bfseries Discontinuous solitons.} Anti-symmetric solitons for which the exciton envelope $\excitonenv(x)$ is discontinuous at its point of symmetry. (See the bullet points above in section~\ref{sec:graphs} for a general explanation.)  The far-field value $\zeta_\infty$ of $\zeta\!=\!g\excitonenv^2$, indicated by the open dots in the leftmost diagram and given by (\ref{zetainfty}), is equal to a double root of the cubic $Q(\zeta)=Q(\zeta;\omega)$ (see (\ref{Q_polynomial})) created by the appropriate choice of constant~$K=K(\omega)$.  In the upper graphs (middle and right), as $x$ travels from $-\infty$ to $\infty$, the pair $(\excitonenv(x),\photonenv(x))$ travels along the cubic from an open circle to a solid dot, then jumps to the other solid dot, and then travels along the cubic to the other open circle.}
\label{fig:band3}
\end{figure}

\begin{figure}[H]
\scalebox{0.51}{\includegraphics{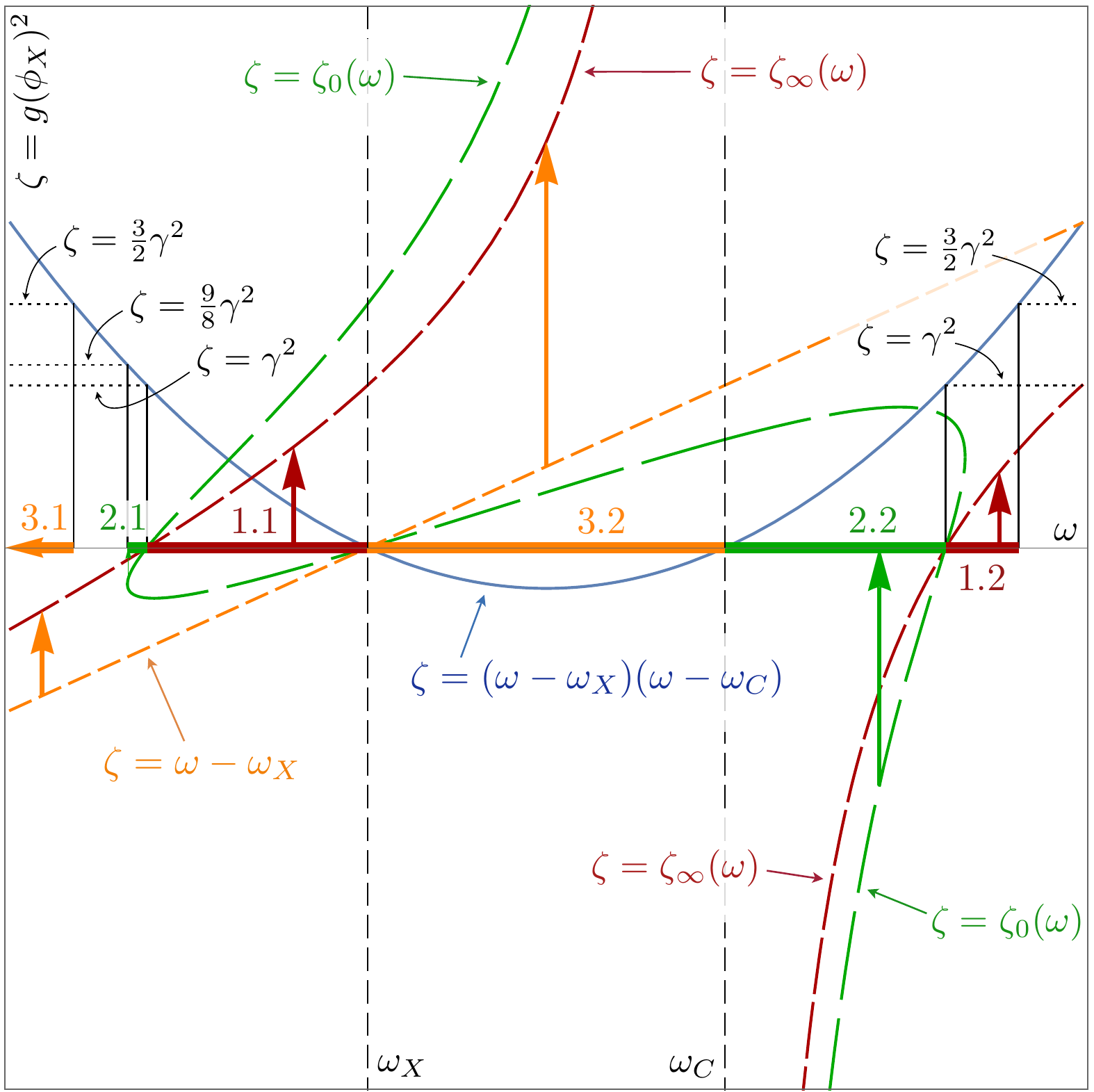}}
\hfill
\scalebox{0.51}{\includegraphics{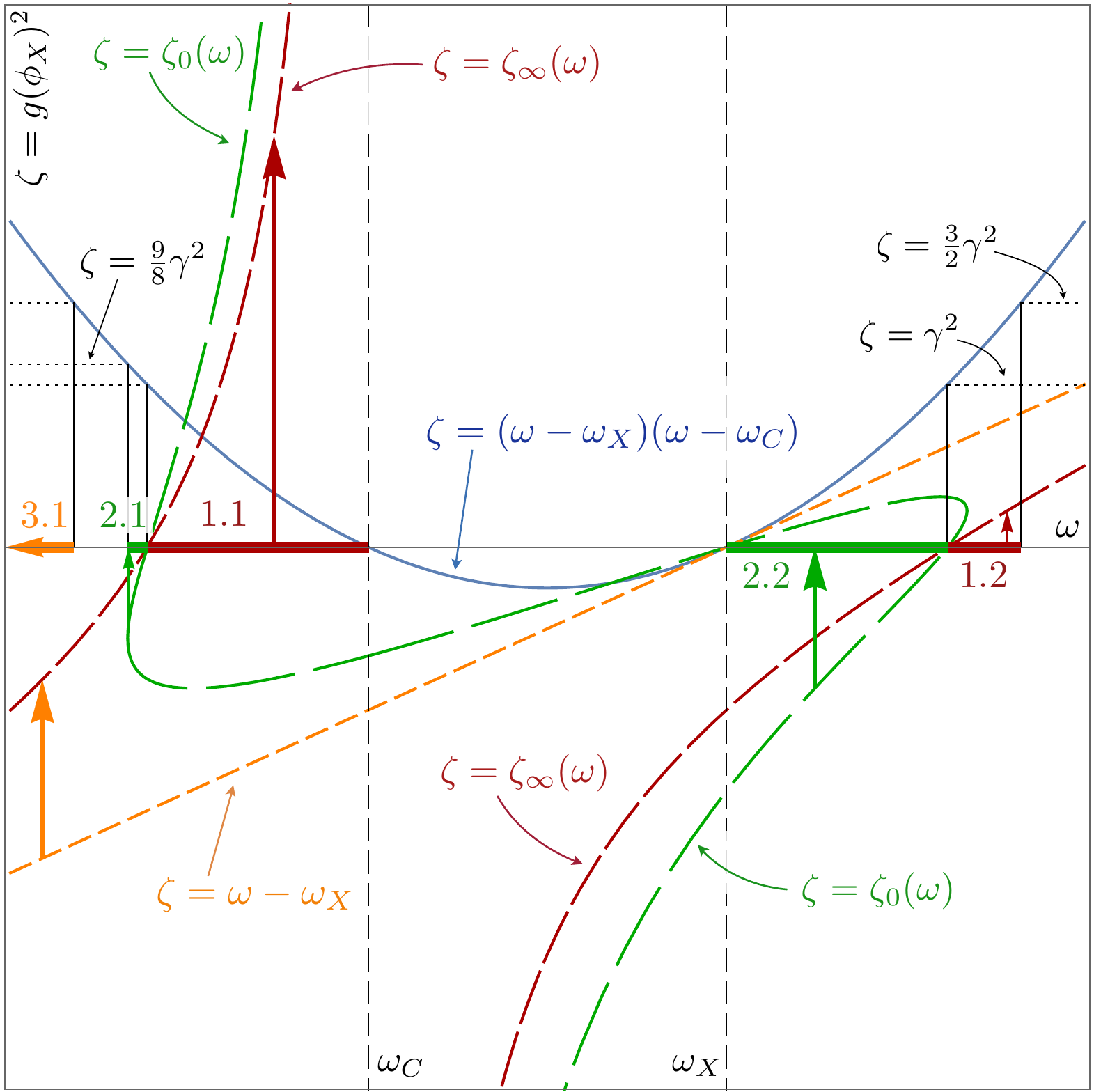}}
\caption{\small {\bfseries Left:} A superposition of the leftmost diagrams of Figures~\ref{fig:band1}, \ref{fig:band2}, and \ref{fig:band3}, showing all bands simultaneously in the positive detuning case, $\oc>\ox$.  A given polariton system admits either bands 1.1, 1.2, and 3.2 if $g>0$ or bands 2.1, 2.2 and 3.1 if $g<0$.
{\bfseries Right:} In the negative detuning case $\oc<\ox$, the band 3.2 of discontinuous solitons is~absent.}
\label{fig:band123}
\end{figure}

\begin{figure}[H]
\centerline{\scalebox{0.58}{\includegraphics{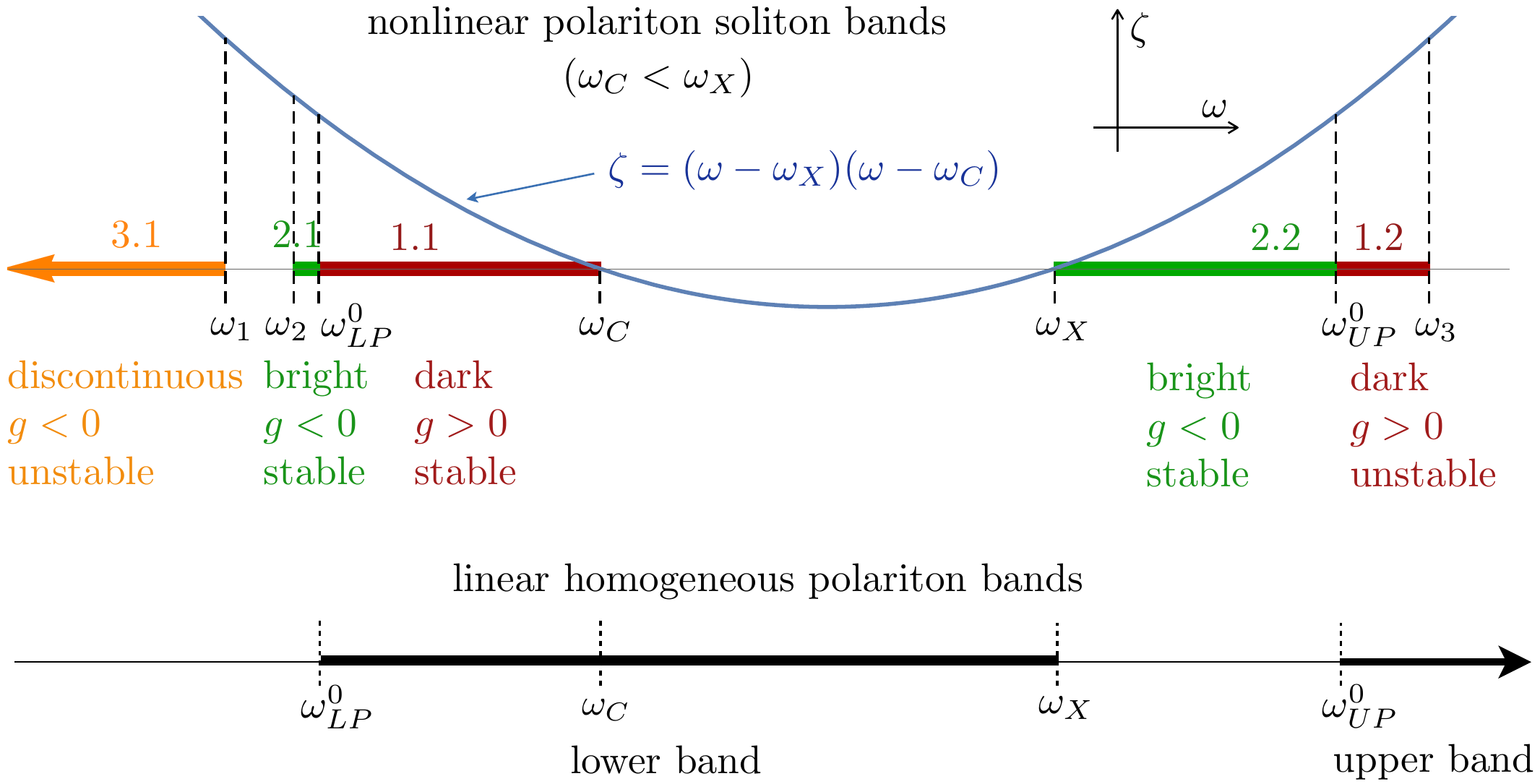}}}
\caption{\small This is the analogous figure to Fig.~\ref{fig:bands} in the negative detuning case $\oc<\ox$.  Note the absence of band 3.2.
}
\label{fig:bandsneg}
\end{figure}

\subsection{A band of continuous and discontinuous dark solitons}\label{subsec:bandD}

Bands 1.1 and 3.2 merge to form a larger band of dark solitons for $g>0$, which we call band~D.  This band was reported by the same authors in \cite{KomineasShipmanVenakides2015}.  It consists of the interval $(\olp,\oc)$, where $\olp$ is defined by
\begin{equation}\label{olp}
  (\olp-\ox)(\olp-\oc)=\gamma^2,
  \qquad
  \olp \,<\, \min\{\ox,\oc\},
\end{equation}
and coincides with the lower endpoint of a well-known band of homogeneous (constant in~$x$) ``lower polaritons" \cite[Fig.~1]{MarchettiSzymanska2011} for the associated linear system obtained by setting $g=0$ and keeping all other parameters unchanged.  The far-field amplitude of the soliton is given by~$\zeta(x)\to\zeta_\infty$ as $|x|\to\infty$, or
\begin{equation}\label{darkamp}
  g\excitonenv^2 \;\;\to\;\; \omega-\ox - \frac{\gamma^2}{\omega-\oc}
  \qquad
  \text{as }\;\; |x|\to\infty,
\end{equation}
and ranges from $0$ to $\infty$ as $\omega$ traverses the band $(\olp,\oc)$.
In the case of negative detuning ($\oc<\ox$), the discontinuous band 3.2 vanishes and the dark soliton is continuous on the entire band~D.  In the case of positive detuning ($\ox<\oc$), the exciton frequency $\ox$ lies within band~D and marks the transition from band 1.1 to band 3.2, where the soliton becomes discontinuous.

Thus, in the positive-detuning case, the frequencies $\olp$, $\ox$, and $\oc$ have the following significance for soliton band~D:

\begin{spslist}
  \item  The value $\olp$ is the {\em threshold frequency} that marks the onset of a soliton.  For frequencies just above this threshold ($0<\omega-\olp\ll1$), the soliton amplitude is small.  This can be seen from equation (\ref{darkamp}), which shows that the soliton amplitude vanishes when $\omega=\olp$.  Thus, solitons at frequencies near the lower edge of the band are in the linear regime because the nonlinearity $g|\excitonenv|^2$ is negligible.
  \item  The exciton frequency $\ox$ is the {\em transition frequency}, at which the exciton field of the soliton becomes discontinuous, as shown in Fig.~\ref{fig:bandD}.  As $\omega$ exceeds $\ox$, the quantity $\wx$ changes from negative to positive and the cubic relation between $\photonenv$ and $\excitonenv$ gains two nonzero roots at $\excitonenv=\pm\phi_0:=\pm\sqrt{\wx/g}\,$ (Fig.~\ref{fig:cubic}, second row).  The exciton field jumps between these two roots exactly when the photon field vanishes, as shown in the rightmost graphs of Fig.~\ref{fig:bandD}.
  \item  The photon frequency $\oc$ is the {\em blowup frequency}: as $\omega$ goes up to $\oc$, the far-field amplitude (\ref{darkamp}) tends to infinity.
\end{spslist}

At the transition frequency $\ox$, the exciton field $\excitonenv(x)$ experiences an infinite slope when its value equals zero, and the graph of $|\excitonenv(x)|^2$ has a cusp (Fig.~\ref{fig:bandDwx}).
This soliton was found numerically in \cite[Fig.~9]{SalasnichMalomedToigo2014}.  Our equations (\ref{darksolution1}--\ref{darksolution2}) give an exact analytic expression of this soliton.

In the negative-detuning case, the threshold and blowup frequencies persist, but the soliton undergoes no transition to discontinuity.

\begin{figure}[H]
\centerline{\scalebox{0.59}{\includegraphics{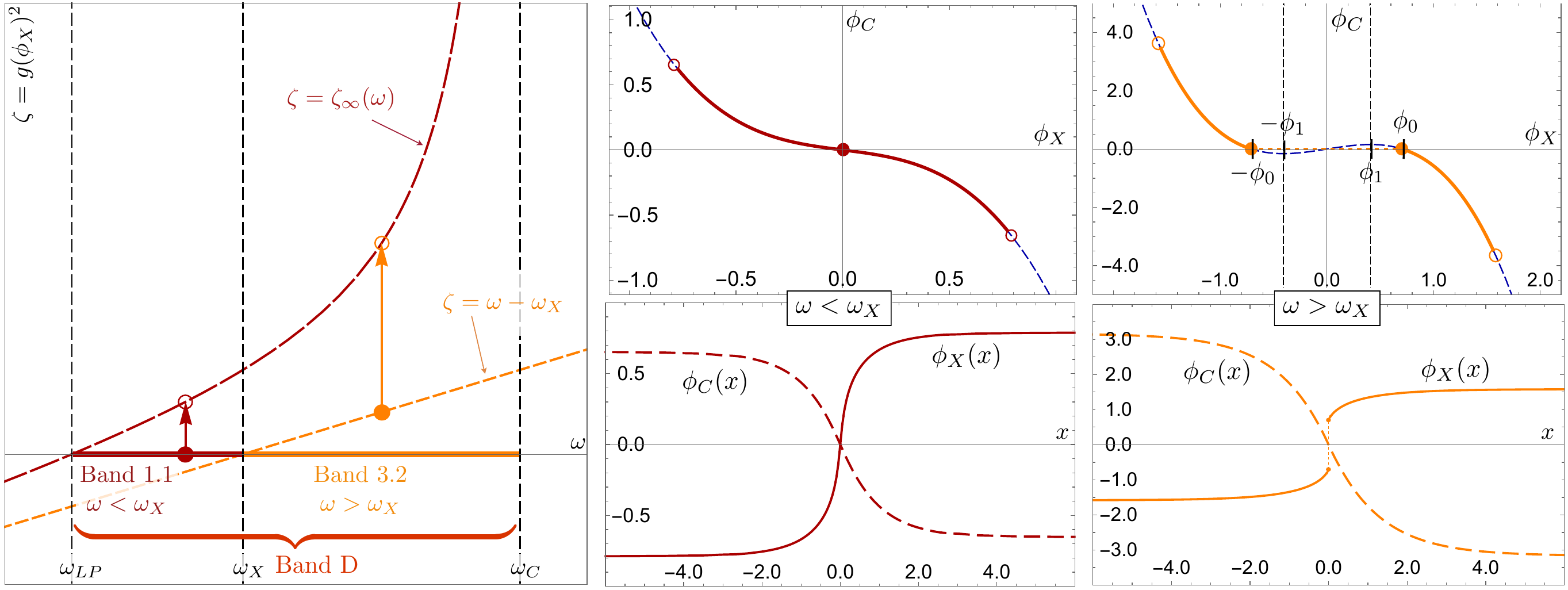}}}
\caption{\small When $\oc>\ox$ (positive detuning), band~1.1 of continuous dark solitons and band~3.2 of discontinuous dark solitons merge to form the single band~D of dark solitons.  (See the bullet points in section~\ref{sec:graphs} for a general explanation.)  Their far-field value of $g\excitonenv^2=\zeta_\infty$ is represented by the single expression $Q(\zeta_\infty)=0$ when the constant $K$ in (\ref{Q_polynomial}) is chosen so that $Q$ has a double root at $\zeta_\infty$; it is given explicitly by (\ref{darkamp}).  The frequency $\olp$ is the threshold frequency, marking the onset of the soliton; the exciton frequency $\ox$ marks the transition from continuous to discontinuous exciton field; and the far-field amplitude of the soliton blows up as $\omega$ goes up to the photon frequency $\oc$.  The values $\pm\phi_0$ are the roots of $\photonenv$ {\itshape vs.} $\excitonenv$, and $\pm\phi_1$ are the critical (local maximum and minimum) points.}
\label{fig:bandD}
\end{figure}

\begin{figure}[H]
\centerline{\scalebox{0.6}{\includegraphics{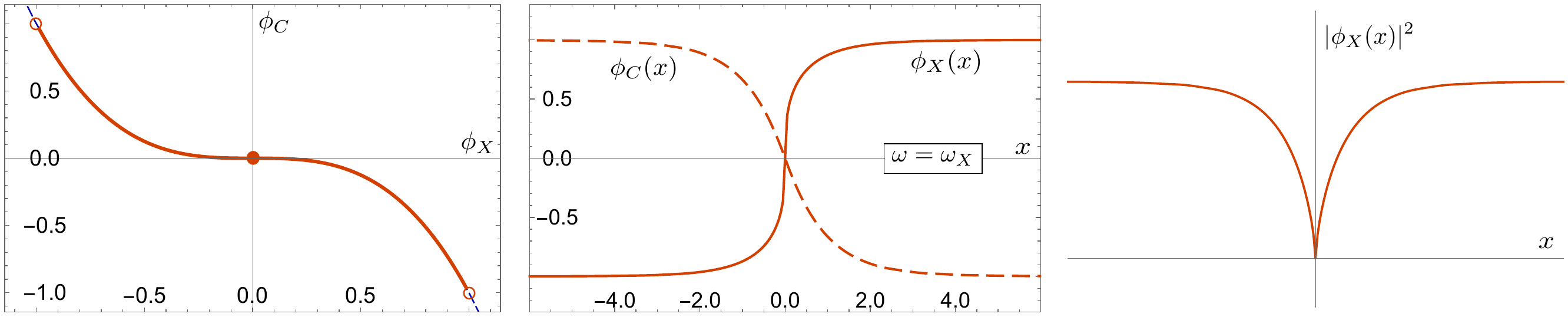}}}
\caption{\small A dark soliton at the transition frequency $\omega=\ox$.  When $\oc>\ox$ (positive detuning) the exciton frequency $\ox$ lies within band D and marks where the exciton field $\excitonenv$ transitions from continuous to discontinuous, as illustrated in Fig.~\ref{fig:bandD}.  The soliton at $\omega=\ox$ is continuous, and the graph of $|\photonenv(x)|^2$ has a cusp at its nadir.}
\label{fig:bandDwx}
\end{figure}

\section{Soliton linear stability}\label{sec:stability}

The polariton  system of equations (\ref{exciton}--\ref{photon}), linearized about a soliton solution, has coefficients that reflect the spatial dependence of the soliton wavefunction.  Thus, an exact linear stability analysis cannot be based on the growth/decay of individual space-harmonic perturbations.  Utilizing the Laplace transform in time, we reduce the soliton linear stability to the invertibility of  a two-component, time-independent  Schr\" odinger operator  in the independent variable $x$.  The operator has  a  matrix potential  $M(x, \tau)$, that carries the soliton information and is required to be invertible for all nonreal values of the parameter $\tau=is$, where $s$ is the Laplace independent variable. The time-symmetry of the problem, arising from the losslessness, makes the classic stability requirement for $\Im \tau>0$ equivalent to the invertibility of the operator $-\partial_{xx} + M(x,\tau)$ for all nonreal $\tau$.   As $|x|$ tends to infinity, the matrix potential $M(x, \tau)$ approaches exponentially an $x$-independent matrix $M(\infty, \tau)$; in order to check the invertibility of the Schr\" odinger  operator, one has to show that it has no  bounded null eigenfunctions.  This is a challenging problem that is currently under study. We present a complete analysis of the linear stability of the soliton far fields in section~\ref{sec:far stability}.

\subsection{Formulation of the stability problem}
 
Inserting $\exciton=\tilde\exciton e^{-i\omega t}$ and $\photon=\tilde\photon e^{-i\omega t}$ into (\ref{exciton},\ref{photon}) yields
\begin{eqnarray}
  i\partial_t\tilde\exciton &=& \left(-\wx + g|\exciton|^2\right)\tilde\exciton + \gamma\tilde\photon\,, \label{excitontilde}\\
  i\partial_t\tilde\photon &=& \left( -\wc - \half \partial_{xx} \right) \tilde\photon + \gamma\tilde\exciton\,. \label{photontilde}
\end{eqnarray}
By taking $\tilde\exciton$ and $\tilde\photon$ to be perturbed envelope functions
\begin{align}
  \tilde\exciton(x,t) &=\excitonenv(x)+\xi_X(x,t)\,, \\
  \tilde\photon(x,t) &=\photonenv(x)+\xi_C(x,t)\,,
\end{align}
one obtains equations for the perturbations $(\xi_X,\xi_C)$,
\begin{eqnarray}
  i\partial_t \xi_X &=& -\wx\xi_X + g\excitonenv^2 \left( 2\xi_X + \bar\xi_X \right) + \gamma \xi_C + \text{h.o.t.}\,,\\
  i\partial_t \xi_C &=& \left( -\wc - \half \partial_{xx} \right) \xi_C + \gamma \xi_X\,,
\end{eqnarray}
in which the omitted terms are higher than linear order in $\xi_X$ and $\bar\xi_X$.
By taking the Laplace transform of these equations and their conjugates, one obtains the system
\begin{eqnarray}
  && \left( is + \wx - 2g\excitonenv^2 \right) \hat \xi_X - g\excitonenv^2\hat{\bar\xi}_X - \gamma\hat\xi_C \,=\, 0\\
  && \left( -is + \wx - 2g\excitonenv^2 \right) \hat{\bar\xi}_X - g\excitonenv^2\hat{\xi}_X - \gamma\hat{\bar\xi}_C \,=\, 0\\
  && \left( is + \wc + \half \partial_{xx} \right) \hat\xi_C - \gamma\hat\xi_X \,=\, 0\\
  && \left( -is + \wc + \half \partial_{xx} \right) \hat{\bar\xi}_C - \gamma\hat{\bar\xi}_X \,=\, 0
\end{eqnarray}
or, in matrix form,
\begin{equation}\label{LinearStability1}
  \renewcommand{\arraystretch}{1.5}
\left[
  \begin{array}{cccc}
    2g\excitonenv^2-\wx-is & g\excitonenv^2 & \gamma & 0 \\
    g\excitonenv^2 & 2g\excitonenv^2-\wx+is & 0 & \gamma \\
    \gamma & 0 & -\partial_{xx}-\wc-is & 0 \\
    0 & \gamma & 0 & -\partial_{xx}-\wc+is
  \end{array}
\right]
\renewcommand{\arraystretch}{1.5}
\left[
  \begin{array}{c}
    \hat\xi_X \\
    \hat{\bar\xi}_X \\
    \hat\xi_C \\
    \hat{\bar\xi}_C
  \end{array}
\right]\,=\,0\,.
\end{equation}
Notice that the field $(\excitonenv,\photonenv)$ occurs in the matrix only through $\zeta(x)=g\excitonenv(x)^2$. 

We write the inhomogeneous version of the linearized problem (\ref{LinearStability1}) as
\begin{equation}\label{LinearStability2}
  \mat{1.2}{A}{\gamma}{\gamma}{B}
  \col{1.2}{\vec{\xi}_X}{\vec{\xi}_C}
  = \col{1.2}{\vec{f_1}}{\vec{f_2}}\,,
\end{equation}
in which
\begin{equation}
  \left[{\vec{\xi}_X},\,{\vec{\xi}_C}\right]^T \,=\, \left[\hat\xi_X,\,\hat{\bar\xi}_X,\,\hat\xi_C,\,\hat{\bar\xi}_C\right]^T
\end{equation}
and $A$ and $B$ are the matrix operators
\begin{equation}
  A = \mat{1.1}{\alpha-\tau}{\zeta}{\zeta}{\alpha+\tau},
  \qquad
  B = \mat{1.1}{-\partial_{xx}-\wc-\tau}{0}{0}{-\partial_{xx}-\wc+\tau}\,,
\end{equation}
where 
\begin{equation}
 \alpha := 2g\excitonenv^2-\wx= 2\z-\wx
\end{equation}
and  $\tau = is$ is a complex frequency.
Equation (\ref{LinearStability2}) can be solved for the photon component:
\begin{equation}\label{LinearStability2a}
  \left(B-\gamma^2 A^{-1}\right) \vec{\xi}_C \,=\, \vec{f_2} - \gamma A^{-1}\vec{f_1}\,.
\end{equation}
The operator on the left is a vector Schr\"odinger operator,
\begin{eqnarray}\label{Schroedinger1}
  B-\gamma^2 A^{-1} &=& -\partial_{xx} - \mat{1.1}{\wc+\tau}{0}{0}{\wc-\tau}
    - \frac{\gamma^2}{\alpha^2-\zeta^2-\tau^2} \mat{1.1}{\alpha+\tau}{-\zeta}{-\zeta}{\alpha-\tau}. 
\end{eqnarray}
Thus, equation \eqref{LinearStability2a} becomes 
\begin{equation}\label{LinearStability2b}
 -(\partial_{xx} + M(x,\tau) \vec{\xi}_C  \,=\, \vec{f_2} - \gamma A^{-1}\vec{f_1}\,,
\end{equation}
where
\begin{equation}
    M(x,\tau) \,=\,  \mat{1.1}{\wc+\tau}{0}{0}{\wc-\tau}
    + \frac{\gamma^2}{\alpha(x)^2-\zeta(x)^2-\tau^2} \mat{1.1}{\alpha(x)+\tau}{-\zeta(x)}{-\zeta(x)}{\alpha(x)-\tau}\,.
\end{equation}
In the definition of $M(x,\tau)$, both $\zeta(x)$ and $\alpha(x)=2\zeta(x)-\wx$ are functions of $x$ that exponentially converge to limiting values as $x\to\pm\infty$.  
Thus the matrix potential $M(x,\tau)$ is an exponentially localized perturbation of its large-$|x|$ value $M(\infty,\tau)$.  

\vskip.3cm

{\em {\bfseries Linear stability} of a solution $(\photon(x),\exciton(x))$, corresponding to $\zeta(x)$ is obtained if, for all $\tau$ in the upper half-plane, the vector Schr\"odinger operator 
$\partial_{xx}+M(x,\tau)$ admits no extended states and no bound states}, that is, bounded vector functions $\vec{\xi}_C(x)$ satisfying the  equation 
\begin{equation}\label{Schroedinger3}
  \left( \partial_{xx} +M(x,\tau) \right) \vec{\xi}_C  \,=\, 0\,.  
\end{equation}

\subsection{Linear stability of the soliton far fields}\label{sec:far stability}

Soliton {\em far-field solutions}, or simply soliton far fields, are  homogeneous ($x$-independent and time-harmonic) solutions of the polariton system of equations (\ref{exciton}--\ref{photon}), that  are  asymptotic to soliton  solutions of the   system as $x$ tends to $\pm\infty$.  They have form  $(\excitonenv^*,\photonenv^*)e^{-i\omega t}$, where $\omega$ is the frequency  and the pair $(\excitonenv^*,\photonenv^*)$ is an equilibrium point of the ODE system (\ref{envelopeODE1a})--(\ref{envelopeODE1b}).

We have seen that the value $\zeta_*$ of the variable $\zeta=g\phi_X^2$ of a soliton far field is independent of  whether $x\to\infty$ or $x\to-\infty$.  The value $\zeta_*$ is a double root of the quartic polynomial $\zeta Q(\zeta)$ in the ODE  \eqref{zetaODE}.  This root takes on one of two values for any given soliton. One value, $\zeta_*=0$ (simple root of $Q(\zeta)$ when $K=0$) is taken by  the solitons of  frequency bands 2.1 and 2.2.   The other value,  $\zeta_*=\zeta_\infty$ (double root of $Q(\zeta)$) is taken  over the frequency bands 1.1, 1.2, 3.1 and  3.2.  Nonzero simple roots of the cubic $Q(\zeta)$ do not correspond to soliton far-field values, and they do not satisfy the original system (\ref{envelopeODE1a}, \ref{envelopeODE1b}).
They are generated from the derivation of the ODE \eqref{zetaODE} for $\zeta$, which  involves multiplying (\ref{envelopeODE1a}) and (\ref{envelopeODE1b}) by the derivatives $\phi_X'$ and $\phi_C'$, which vanish when $\phi_X$ and $\phi_C$ are constant. 

For soliton far fields, the Schr\" odinger equation \eqref{Schroedinger3} reduces to the constant-coefficient problem
\begin{equation}\label{Schroedinger4}
  \left( \partial_{xx} +M(\infty,\tau) \right) \vec{\xi}_C(x) \,=\, 0\,,
\end{equation}
by putting $\zeta(x)=\zeta_*$ and $\alpha(x)=\alpha_*$.
The {\em far-field dispersion relation}, relating wave number $k$ and  frequency $\tau$ as $|x|\to\infty$   for the linearization (\ref{Schroedinger4}), is obtained by replacing $\partial_{xx}$ with $-k^2$ and setting the determinant to zero,
\begin{multline}\label{dispersion2}
 D(k,\tau)= \det\left(k^2 I - M(\infty,\tau)\right) \\
 = \tau^4-\left(\alpha_*^2+\bsq+2\gsq-\zeta_*^2\right)\tau^2+\left(\g^4-2\alpha_*\bb\gsq+\bsq(\alpha_*^2-\zeta_*^2)\right)  = 0\,
\end{multline}
in which $\beta = k^2-\wc$ and $\alpha_*=2\zeta_*-\wx$.  $D(k,\tau)$ is a function of $k^2$ and $\tau^2$.

The stability condition for homogeneous solutions obtained in the previous section can be rephrased as follows: For all $\tau\in\mathbb{C}\setminus\mathbb{R}$, {\em the matrix $M(\infty,\tau)$ has no eigenvalues $k^2$ in $[0,\infty)$, or, equivalently, $D(k,\tau)\not=0$ for such $\tau$ and $k$.}

The proof of the following theorem is given in section~\ref{sec:proofstability}.

\begin{theorem}\label{thm:stability}
Let $(\exciton(x,t),\photon(x,t))=(\excitonenv(x),\photonenv(x))e^{-i\omega t}$ be a solution of the nonlinear polariton system (\ref{exciton},\ref{photon}) such that $(\excitonenv(x),\photonenv(x))$ has limits as $x\to\pm\infty$.  All such solutions are described in Theorem~\ref{thm:solitons}, and one has
\begin{equation*}
  (\excitonenv(x),\photonenv(x)) \to \pm(\excitonenv^*,\photonenv^*)
  \quad\text{as}\quad
   x\to\pm\infty.
\end{equation*}
The function pair $(\excitonenv^*,\photonenv^*)e^{-i\omega t}$ is a homogeneous solution to (\ref{exciton},\ref{photon}) that is linearly
\begin{eqnarray*}
  \text{stable} &\text{if}& \omega\in \ \text{band 1.1, 2.1, 2.2, or 3.2,}\\
  \text{unstable} &\text{if}& \omega\in \ \text{band 1.2 or 3.1.}
\end{eqnarray*}
In bands 2.1 and 2.2, $(\excitonenv^*,\photonenv^*)=(0,0)$, and in the other bands, $(\excitonenv^*,\photonenv^*)=\pm(\excitonenv^\infty,\photonenv^\infty)$, defined in (\ref{eqsoln}).
\end{theorem}
\medskip

{\it Remark:}  The determinant \eqref{dispersion2} can be obtained directly (as a function of $s^2$ as opposed to the above $\tau^2=-s^2$) as the determinant of the matrix \eqref{LinearStability1}, in which $\partial_{xx}$ is replaced by $-k^2$. The stability condition, rephrased for the $s^2$ variable is: {\em Linear stability at all modes $k$ requires that the two roots of the determinant $D$, considered as a quadratic in the variable $\ssq$, be negative or zero for all real values of $k$, {\it i.e.}  for all values of $\bb$ that satisfy $\bb\ge-\wc$.} We use this approach in the proof in section~\ref{sec:proofstability}.

\section{Proof of Theorem~\ref{thm:solitons}: derivation of solitons}\label{sec:proofsolitons}

This section contains the proofs of the six classes of solitons described in Theorem~\ref{thm:solitons} of section~\ref{sec:solitons}.

The derivation of these solitons is simplified by passing to a normalized frequency variable $\eta$ that conveniently parameterizes the operating frequency $\omega$,
\begin{equation}\label{definition_eta}
 \eta=\textstyle\frac{1}{\gsq}(\omega-\ox)(\omega-\oc)=\textstyle\frac{\wx\wc}{\gsq}. 
\end{equation}
This expression, which is quadratic in $\omega$, produces generically two frequencies for the same value of $\eta$, a first indication of the fact  that exciton-polariton soliton solutions typically come in pairs.  One only needs to consider values
\begin{equation}
 \eta\;\ge\;\eta_\text{min}= \textstyle\frac{(\ox-\oc)^2}{4\gsq}, 
\end{equation}
that are at or above the minimum $\eta_\text{min}$ of the quadratic and thus produce real frequencies.

The non-dimensional frequency $\wx$ is a natural scaling factor for $\zeta$:
\begin{equation}
\z= \wx u\,, \ \ \ \
\z_\infty=\wx u_\infty\,, \ \ \ \  
u_\infty=1-\frac{1}{\eta}\,.
\end{equation}
Using $u$ instead of $\zeta$ greatly simplifies the algebraic computations of solitons.
In the new variables, the ODE (\ref{zetaODE}) for $\z$ becomes an ODE for $u$,
\begin{equation}\label{ode_u}
\begin{split}
 (3u-1)^2u'^2 &=-8\wc u \left[u^3-\left(2-\textstyle\frac{3}{2\eta}\right)u^2+\left(1-\textstyle\frac{1}{\eta}\right)u +\tilde K\right],  \\
   &= -8\wc u\, B(u)
   \qquad \mbox{when $\wx\ne 0$}.
\end{split}
\end{equation}
The polynomial $B(u)$ is related to $Q(\zeta)$ (see~(\ref{Q_polynomial})) by
\begin{equation}
  B(u) = \frac{-1}{\wc\wx^3}\, Q(\wx u)\,.
\end{equation}
The phase space of the ODE (\ref{ode_u}) is one-dimensional (the $u$-axis).
Because $u'$ appears squared, a soliton solution of the ODE is obtained through standard phaseline analysis by {\em connecting a double root of the quartic polynomial $uB(u)$ on the right side of the ODE with a simple root of $uB(u)$.}  A solution connecting these consecutive roots is possible provided that the interval between the two roots does not contain the singular value $u=1/3$; we call this the {\bfseries\itshape connectivity condition}.  In addition, the {\bfseries\itshape sign condition} requires the positivity of the right side of the equation over the interval between the two roots; it can be expressed~as
\begin{equation}\label{sign_condition}
 -\wc(uB(u))''|_{u=\text{double root}}\ge 0.
 \qquad
 \text{(sign condition)}
\end{equation}

The solution $u(x)$ approaches the double root exponentially slowly as $x$ approaches $\pm\infty$; thus the double root signifies the far-field amplitude of the soliton and is thus equal to $u_\infty$.  The simple root $u_0$ is the extremal value (maximum or minimum) of $u(x)$ and is attained at a finite value $x_*$.  The solution $u(x)$ is symmetric about $x_*$, and has local quadratic behavior there.
Because of the invariance of the polariton system under a shift $x\mapsto x-x_*$, we will henceforth take $x_*=0$.

Since $uB(u)$ vanishes at $u=0$, it is guaranteed that the interval between two consecutive roots is either positive or negative, so that the solution $u(x)$ is of one sign.  This allows one to choose the sign of $g$ appropriately so that
$(\wx/g)u>0$ and $\excitonenv=\pm\sqrt{(\wx/g)u\,}$.  If the simple root $u_0$ is nonzero, then the exciton envelope field is symmetric and of one sign,
\begin{equation*}
  \excitonenv(x) \,=\, \pm\sqrt{\frac{\wx}{g}\,u(x)}\,,
  \qquad
  \text{if }\;\; u_0 \not= 0\,.
\end{equation*}
If the simple root $u_0$ is equal to zero, then the exciton envelope field is anti-symmetric,
\begin{equation*}
  \excitonenv(x) \,=\, \pm\, \mathrm{sgn}(x) \sqrt{\frac{\wx}{g}\,u(x)}\,,
  \qquad
  \text{if }\;\; u_0 = 0\,.
\end{equation*}

\bigskip

\noindent
{\bfseries 1. Dark solitons.}  These are solutions that connect a double nonzero root of $uB(u)$, (far-field value) with the zero root (nadir).

The following factorization is key to the analysis;
 \begin{equation}\label{factorization}
B(u)=(u-u_\infty)^2(u-u_0), \ \ \ \ \begin{cases}u_\infty=1-\frac{1}{\eta}\\ u_0=\frac{1}{2\eta}\end{cases},
 \ \ \ \ \tilde K=-u^2_\infty u_0.\\
\end{equation}
One observes that the singular value $u=1/3$ of the ODE lies between the double root $u_\infty$ and the nonzero simple root $u_0$ and thus obstructs a soliton connection between them:
\begin{equation}
 \begin{cases} u_\infty
 =\frac{1}{3}+2\left(\frac{1}{3}-\frac{1}{2\eta}\right),\\
u_0=\frac{1}{3}-\left(\frac{1}{3}-\frac{1}{2\eta}\right).
 \end{cases}
\end{equation}
In order for $u_\infty$ to be connectible to the other root $u=0$ of $uB(u)$, the following two mutually equivalent conditions must hold:
\begin{equation}\label{connectivity1}
 u_\infty<\textstyle\frac{1}{3}, \quad\text{i.e.,}\quad 0<\eta<\textstyle\frac{3}{2}.
\qquad
\text{(connectivity condition)}
\end{equation}
Furthermore, one obtains from  \eqref{sign_condition}, the sign condition
\begin{equation}\label{sign_condition_red}
 -\wc u_\infty(u_\infty-\textstyle\frac{1}{3})>0.
 \qquad
 \text{(sign condition)}
\end{equation}
As a result of the positivity of $\eta$, the frequencies $\wx$ and $\wc$ must have the same sign. There are therefore two cases $(\wx<0, \ \wc<0)$ and $(\wx>0, \ \wc>0)$.

In the case that $\wx$ and $\wc$ are both negative, or $\omega<\min\{\ox, \oc\}$, conditions (\ref{connectivity1}, \ref{sign_condition_red}) necessitate $u_\infty<0$, which by $u_\infty=1-\eta^{-1}$ is equivalent to $0<\eta<1$.  This yields the band
\begin{equation}
  \{ 0<\eta<1, \;\, \omega<\min\{\ox, \oc\}\}.  \qquad\text{(band 1.1)}
\end{equation}
Thus for each $\eta$ between $0$ and $1$ the ODE (\ref{ode_u}) has a soliton solution $u(x)$ with nadir equal to the simple root $u=0$ and far-field value equal to the double root $u_\infty$.

When converting $\eta$ back to the variable $\omega$ through $\eta=\wx\wc/\gamma^2$, the condition $\omega<\min\{\ox, \oc\}$ determines the choice of frequency $\omega$ as the lower of the two solutions of
$\eta\gamma^2=(\omega-\ox)(\omega-\oc)$.  This results in band 1.1 stated in Theorem~\ref{thm:solitons} (see \ref{bands1}) and depicted in Fig.~\ref{fig:band1}.  The interval $0<\eta<1$ corresponds to the condition $0<\wx\wc<\gamma^2$ in~(\ref{bands1}).

In converting $u$ back to the variable $\zeta$, notice that $u_\infty$ and $\wx$ are both negative so that $\zeta_\infty=\wx u_\infty>0$, and thus $\zeta(x)=\wx u(x)>0$ for all $x$ since $u$ is of one sign on the interval~$(u_\infty,0)$.
The sign of $g$ is determined by $0<\zeta=g\excitonenv^2$, so $g>0$ for these solutions.  The far-field value (\ref{farfield1}) is obtained from the expression (\ref{zetainfty}) for $\zeta_\infty$.

In the case that $\wx$ and $\wc$ are both positive, or $\omega> \max\{\ox, \oc\}$,
conditions (\ref{connectivity1}, \ref{sign_condition_red}) necessitate $0<u_\infty<\frac{1}{3}$, which by $u_\infty=1-\eta^{-1}$ is equivalent to $1<\eta<\frac{3}{2}$.   This yields the band
\begin{equation}
  \{ 1<\eta<\textstyle\frac{3}{2}, \;\, \omega>\max\{\ox, \oc\}\}.  \qquad\text{(band 1.2)}
\end{equation}
This results in band 1.2 stated in Theorem~\ref{thm:solitons} (see \ref{bands1}) and depicted in Fig.~\ref{fig:band1}.  The interval $1<\eta<\frac{3}{2}$ corresponds to the condition $\gamma^2<\wx\wc<\frac{3}{2}\gamma^2$ in~(\ref{bands1}).
Since again $\zeta_\infty=\wx u_\infty>0$, one has $g\excitonenv(x)^2=\zeta(x)>0$ for all $x$, and again $g>0$.

\bigskip

\noindent
{\bfseries 2. Bright solitons.}  These are solutions that connect zero as a double root of $uB(u)$ (far-field value) to a simple root (peak).

Assuming $uB(u)$ has a double root at $0$, \eqref{ode_u} takes the form
\begin{equation}\label{ode_u2}
 (3u-1)^2u'^2=-8\wc u^2\underbrace{\left[u^2-\left(2-\textstyle\frac{3}{2\eta}\right)u+1-\textstyle\frac{1}{\eta}
\right]}_{{\text{Quadratic }} A(u)},  \ \ \ \text{when $\wx\ne 0$}.
\end{equation}
We are interested in the ranges of $\eta$ for which roots of the quadratic $A(u)$  are to the left of the singular point $u=\textstyle\frac{1}{3}$ and thus can be connected with the double root at $u=0$ (connectivity condition). 
By a simple argument,$\footnote{\mbox{Write 
$ (u-1)^2=-\textstyle\frac{3}{2\eta}(u-\textstyle\frac{2}{3})$ and examine how the line cuts the quadratic, as $\eta$ is varied.}}$ there is either one root $u_0$ or no root in the  half-line $u<\textstyle\frac{1}{3}$. The range of $\eta$ for which this root is present is 
\begin{equation}
 0<\eta<\textstyle\frac{9}{8}, \ \ \ \
\begin{cases}0<u_0<\textstyle\frac{1}{3}, 
\ \ \ \ \ \text{when $1<\eta<\textstyle\frac{9}{8}$}, \\
u_0<0, \ \ \ \ \ \ \ \text{when $0<\eta<1$}.
\end{cases}
\end{equation}
The other root is above $\frac{1}{3}$, so one computes
\begin{equation}\label{u0}
  \textstyle u_0 = 1-\frac{1}{\eta}\left[ \frac{3}{4} + \sqrt{\frac{9}{16} - \frac{\eta}{2}\,}\, \right].
\end{equation}
As $\eta$ decreases from $\eta=\textstyle\frac{9}{8}$ to  $\eta=1$, the root $u_0$ descends from $u=\textstyle\frac{1}{3}$ to $u=0$, then turning negative as $\eta$ decreases from the value $1$. In the limit $\eta\to +0$, $u_0\to -\infty$. 

The condition $\eta>0$ implies $\wx\wc>0$, so that, as before, either $(\wx<0, \ \wc<0$) or $(\wx>0, \ \wc>0)$.  We must consider these cases in conjunction with the sign condition discussed above, which requires the right side of \eqref{ode_u2} to be positive in the neighborhood of the double root $u=0$,
\begin{equation}\label{sign_condition2}
 \wc(1-\textstyle\frac{1}{\eta})<0. 
 \qquad
 \text{(sign condition)}
\end{equation}
Putting these requirements together, we obtain two soliton bands,
\begin{equation}
 \begin{cases}
 \text{band 2.1}
=\{ 1<\eta<\textstyle\frac{9}{8}, \  \wx<0, \ \wc<0 \}\\ 
  u_0>0, \ \ \ \ \z_0=\wx u_0<0, \ \ \ \ g<0,
\end{cases}
\end{equation}
\begin{equation}
\begin{cases}
\text{band 2.2}=\{0<\eta<1,  \  \  \wx>0, \ \wc>0 \}\\
  u_0<0, \ \ \ \ \z_0=\wx u_0<0, \ \ \ \ g<0.
  \end{cases}
\end{equation}
As in the previous case of dark solitons, the endpoints of the frequency bands are imposed by the bounds of $\eta$ and the relation $\eta\gamma^2=\wx\wc$, and the sign of $g$ coincides with the sign of~$\zeta$, which is negative in these cases.

The minimal (negative) value $\zeta_0$ of $\zeta(x)=\wx u(x)$ is equal to $\wx u_0$, which from (\ref{u0}) is equal~to
\begin{equation}
  \textstyle
  \zeta_0 \,:=\, \wx u_0 \,=\, \wx - \frac{\gamma^2}{\wc}
                 \left[ \frac{3}{4} + \sqrt{\frac{9}{16} - \frac{\wx\wc}{2\gamma^2}\,}\, \right],
\end{equation}
and since $g<0$, one obtains the peak value of $|g|\excitonenv^2$ stated in the theorem.

\bigskip
\noindent
{\bfseries 3. Discontinuous solitons.}\label{discontinuoussolitons}
In the derivation of the dark solitons, the case $u_\infty>\textstyle\frac{1}{3}$ was excluded because the  connectivity of $u_\infty$ to the zero root was broken by the singularity at $u=\textstyle\frac{1}{3}$.   Consider instead the $u$-interval between $u_\infty$ and $u=1$, which does not contain $1/3$ or any root (besides $u_\infty$) of $B(u)$.  The corresponding $\zeta$-interval connects $\zeta_\infty$ an $\wx$ and does not contain $\wx/3$ or any other roots (besides $\zeta_\infty$) of $\zeta Q(\zeta)$.

This $\zeta$-interval corresponds to two $\excitonenv$-intervals, connecting $\pm\excitonenv^\infty$ with $\pm\phi_0:=\pm\sqrt{\wx/g}$ and not containing $\pm\sqrt{\wx/(3g)\,}$, where $\excitonenv^\infty>0$ is defined through $g(\excitonenv^\infty)^2=\zeta_\infty$.  Naturally, $g$ must take the sign of $\zeta_\infty$, and thus the cubic (\ref{phirelation}) giving $\photonenv$ as a function of $\excitonenv$ vanishes when $\excitonenv=\pm\phi_0$. 

Given that the sign condition holds, a discontinuous soliton is constructed by taking a solution of (\ref{zetaODE}) for which $\zeta(x)$ travels from $\phi_0$ to $\excitonenv^\infty$ as $x$ travels from $0$ to $\infty$, then setting
\begin{eqnarray*}
  \excitonenv(x) = \sqrt{\zeta(x)/g\,} &\text{for}& x>0 \\
  \excitonenv(x) = -\sqrt{\zeta(-x)/g\,} &\text{for}& x<0 \\
  \photonenv(x) =  \gamma^{-1}\excitonenv\left(\wx-g\excitonenv^2\right) &\text{for}& x\in\mathbb{R}.
\end{eqnarray*}
The polariton field $\left(\excitonenv(x),\photonenv(x)\right)$ is antisymmetric about $x=0$ and satisfies the pair (\ref{envelopeODE1a},\ref{envelopeODE1b}) for $x\not=0$. 
Since $\photonenv$ vanishes when $\excitonenv=\pm\phi_0$, setting $\photonenv(0)=0$ makes the field $\photonenv(x)$ continuous; this together with antisymmetry makes $\photonenv(x)$ continuously differentiable at $x=0$.  Thus (\ref{envelopeODE1b}) is satisfied in the sense of distributions, even through $x=0$, and the jump of $\photonenv''(x)$ across $x=0$ is computed from the ODE:
\begin{equation}\label{jump}
  [\photonenv''(x)]_{x=0} \,=\, 2\gamma [\excitonenv(x)]_{x=0} \,=\, 4\gamma \sqrt{\wx/3\,}\,.
\end{equation}

Violation of the connectivity condition means
\begin{equation}\label{noconnectivity1}
  u_\infty>\textstyle\frac{1}{3},
  \quad\text{i.e.,}\quad
  \eta>\textstyle\frac{3}{2} \;\;\text{or}\;\; \eta<0.
\qquad
\text{(no-connectivity condition)}
\end{equation}
The sign condition \eqref{sign_condition_red} still applies and reduces to 
\begin{equation}
 \wc<0,
\end{equation}
as a result of $u_\infty>\textstyle\frac{1}{3}$. The sign of $\wx$ is opposite to the sign of $\eta$, as follows from the definition  of $\eta$ \eqref{definition_eta}.
From the relation $u_\infty=1-\eta^{-1}$, we obtain two frequency bands, one for $\eta>3/2$, and one for $\eta<0$.  For $\eta>3/2$,
\begin{equation}
 \begin{cases} \mbox{band 3.1}
=\{\eta>\textstyle\frac{3}{2}, \  \omega<\min\{\ox, \oc\} \}\\ 
  u_\infty>\textstyle\frac{1}{3}, \ \ \ \z_\infty=\wx u_\infty<0, \ \ \ \ g<0.
\end{cases}
\end{equation}
In this band, the far-field value of $u$ is $u_\infty=1-\eta^{-1}$, so the range of $u(x)$ is $(u_\infty,1)$ and thus $\zeta(x)=\wx u(x)$ has a far-field value of $\wx-\gamma^2/\wc$ and, since $\wx<0$, its range is equal to negative interval $(\wx,\zeta_\infty)$.  The corresponding discontinuous soliton is bright since $|\zeta_\infty|<|\wx|$.

In the case $\eta<0$, one obtains
\begin{equation}
 \begin{cases} \mbox{band 3.2}
=\{\eta<0, \ \ox< \omega<\oc \}\\ 
  u_\infty>\textstyle\frac{1}{3}, \ \ \ \z_\infty=\wx u_\infty>0, \ \ \ \ g>0.
\end{cases}
\end{equation}
The far-field amplitude of $u$ is again $u_\infty=1-\eta^{-1}$, so the range of $u(x)$ is $(1,u_\infty)$  Since $\wx>0$, the range of $\zeta(x)$ is the positive interval $(\wx,\zeta_\infty)$.

\section{Proof of Theorem~\ref{thm:stability}: far-field stability}\label{sec:proofstability}

This section is devoted to a proof of Theorem~\ref{thm:stability}, using the notation introduced there.
The determinant $D(\infty,\tau)$ (\ref{dispersion2}) with $\tau=is$ is
\begin{equation}\label{dispersion1}
 D \,:=\, s^4+(\asq+\bsq+2\gsq-\zsq)\ssq+\g^4-2\aaa\bb\gsq+\asq\bsq-\zsq\bsq,
\end{equation}
with $\zeta=0$ or $\zeta=\zeta_\infty$.
{\em Linear stability at all modes $k$ requires that the two roots of the determinant $D$, considered as a quadratic in the variable $\ssq$, be negative or zero for all real values of $k$, {\it i.e.}  for all values of $\bb$ that satisfy $\bb\ge-\wc$.}
This is equivalent to the following three conditions:
\begin{enumerate}
 \item The product of the roots is positive or zero
\begin{equation}\label{positive_root_product}
  \g^4-2\aaa\bb\gsq+\asq\bsq-\zsq\bsq\ge 0, \ \ \ \text{for all $\bb\ge-\wc$.}
\end{equation}
\item Their sum  of the roots is negative or zero
\begin{equation}\label{negative_root_sum}
 \asq+\bsq+2\gsq-\zsq\ge 0, \ \ \ \text{for all $\bb\ge-\wc$.}
\end{equation}
\item The discriminant is positive or zero
\begin{equation}\label{positive_discriminant}
 (\asq+\bsq+2\gsq-\zsq)^2- 4(\g^4-2\aaa\bb\gsq+\asq\bsq-\zsq\bsq)\ge 0, \ \ \ \text{for all $\bb\ge-\wc$.}
\end{equation}
\end{enumerate}
Inequality \eqref{positive_discriminant}  is the hardest of the three conditions to analyze. Through algebraic manipulation, it is recast as
\begin{equation}\label{stability_condition_3a}
(\asq-\bsq-\zsq)^2+4\gsq(\asq+\bsq+2\aaa\bb-\zsq)\ge0,    \ \ \ \text{for all $\bb\ge-\wc$}.
\end{equation}

The three inequalities together constitute {\it necessary and sufficient conditions} for the asymptotic values  $\z=0$ or $\z=\z_\infty$ of a soliton solution to be a  {\it linearly stable} homogeneous solution. We refer to these inequalities  below as the first, second, and third stability conditions. 

\subsection{Stability of the far-field solution $\z=0$.} 

The left side of each of the three inequalities above is  either a perfect square or a sum of squares (see the third inequality in its recast form~(\ref{stability_condition_3a})).  Thus, they are all satisfied, and so the soliton far-field solutions for bands 2.1 and 2.2 are stable.

\subsection{Stability of the far-field solution $\z=\z_\infty$.}

\subsubsection*{First stability condition.} 

For $\z=\z_\infty$, the left side of the inequality \eqref{positive_root_product} factors to 
\begin{equation}
 \left( \gsq-(\aaa_\infty+\z_\infty)\bb \right) \left( \gsq-(\aaa_\infty-\z_\infty)\bb \right) \ge 0, 
\end{equation}
in which $\aaa_\infty$ is the value of $\aaa$ at $\z=\z_\infty$. 
The definition of $\eta$ gives directly 
\begin{equation}\label{homogeneous_values}
 \wx=\textstyle\frac{\eta\gsq}{\wc}, 
\end{equation}
from which one obtains easily
\begin{equation}
\ \ \ \aaa_\infty+\z_\infty=\textstyle\frac{\gsq}{\wc}(2\eta-3),
\ \ \ \aaa_\infty-\z_\infty=-\textstyle\frac{\gsq}{\wc}.  
\end{equation}
Inserting these into the above inequality and recalling that  $\bb=k^2-\wc$, yields 
\begin{equation}\label{alpha_pm_beta}
\left[\left(\eta-\textstyle\frac{3}{2}\right) \textstyle\frac{k^2}{\wc}+\eta-1)\right]\textstyle\frac{k^2}{\wc}\ge 0 \ \ \ \ \ \ \text{for all real $k$.}
\end{equation}
The sign distribution of the left of the inequality 
reveals that the inequality is satisfied in exactly two regimes
\begin{equation}\label{stability1_regimes_eta}
 \begin{cases} \eta_\text{min}\le\eta\le 1, \ \ \  \wc<0 \\
 1\le\eta\le\textstyle\frac{3}{2}, \ \ \ \ \ \ \  \wc>0.
 \end{cases}
\end{equation}

The homogeneous solutions corresponding to the far-field values of the solitons in bands 1.1 and 3.2 are in the first regime.  Those corresponding to band 1.2 are in the second regime. Thus, the far-field values of all these dark solitons pass the first test for linear stability.  On the other hand, the homogeneous solutions corresponding to the far-field values of the solitons in band 3.1 are outside these two regimes and therefore are not linearly stable.

\medskip

The first stability condition \eqref{positive_root_product} poses a simple restriction  on the homogeneous solutions  $\z=\z_\infty$, as one observes that it  is a quadratic inequality in the variable $\beta/\gsq$,
\begin{equation}\label{beta_quadratic}
(\asq-\zsq)\Biggl(\frac{\bb}{\gsq}\Biggr)^2
-2\aaa \Biggl(\frac{\bb}{\gsq}\Biggr)+1\ge 0,  \ \ \ \text{for all $\bb\ge-\wc$.}
\end{equation}
Necessarily, 
\begin{equation}\label{alphasqgezsq}
 \asq-\zsq\ge0. 
\end{equation}
This is an interesting inequality. 
Factoring and recalling that $\aaa=2\z-\wx$,  it becomes,
\begin{equation}\label{rootcondition1}
 (\zeta-\wx)(\z- \frac{1}{3}\wx)\ge0
\end{equation}
Thus, stable homogeneous solutions takes  values $\z=\z_\infty$ that {\it  lie outside the open interval between $\wx$ and $\wx/3$}. To the {\it right} of this interval $\aaa>0$, while   $\aaa<0$ holds when $\z$ is to the {\it left} of the interval.

\subsubsection*{Second stability condition.} 

Inequality in  \eqref{negative_root_sum} follows immediately from the obtained  requirement of the first stability condition, $\asq-\zsq\ge0$.

\subsubsection*{Third stability condition.}

The far-field solutions for bands 1.1 and 3.2 satisfy the first two stability conditions.  We show now that they also satisfy the third condition.  It suffices to show that the second term in parentheses (call it $A$) in \eqref{stability_condition_3a} is positive or zero.  The proof is based on the fact that both solutions have $\wc<0$.
\begin{equation}
A= (\aaa+\bb)^2-\zsq=(\aaa+\bb+\z)(\aaa+\bb-\z). 
\end{equation}
Inserting the expression \eqref{alpha_pm_beta} for $\aaa\pm\bb$, we obtain 
\begin{equation}
 A=(\textstyle\frac{\gsq}{\wc}(2\eta-3)+k^2-\wc)(-\textstyle\frac{\gsq}{\wc}+k^2-\wc). 
\end{equation}
The term $2\eta-3$ is positive or zero by \eqref{stability1_regimes_eta}. With $\wc<0$, every term in each of the two parenthesis is positive or zero.

The expression for $A$ above is quadratic in $k^2$ with roots $-\textstyle\frac{\gsq}{\wc}(2\eta-3)+\wc$ and $\textstyle\frac{\gsq}{\wc}+\wc$.   For the far-field solutions for band 1.2, necessarily $\wc>0$, and thus both roots are positive.  Giving $k^2$ a value between these roots makes the quadratic expression negative.  The third stability condition is thus violated, so the far-field solution for band 1.2 is unstable.  Table~\ref{table:stability} gives a summary of linear far-field stability of all solitons.

\begin{center}
\begin{minipage}{0.7\textwidth}
\captionof{table}{Soliton properties}\label{stab_table_z_pos}
\begin{tabular}{| c | l | c | c |}
\hline
Band &  Bright/Dark &  Linearly stable far field & $\eta$ domain
\\ \hline\hline
3.1 & Neither & no & $\hspace{1.6em} \frac{3}{2}<\eta<\infty  $
\\ \hline
2.1 &  Bright & yes & $\hspace{1.3em} 1<\eta<\textstyle\frac{9}{8}$
\\ \hline
1.1 & Dark & yes & $\hspace{1.3em} 0<\eta<1$ 
\\ \hline
3.2  & Dark  & yes &  $\eta_\text{min}<\eta<0$
\\ \hline
2.2  & Bright  & yes &  $\hspace{1.3em} 0<\eta<1$
\\ \hline
1.2  & Dark & no &  $\hspace{1.3em} 1<\eta<\textstyle\frac{3}{2}$
\\ \hline
\end{tabular}
\label{table:stability}
\smallskip
\end{minipage}
\end{center}

\section{Concluding Discussion}

We have studied soliton solutions in a polariton condensate and have derived the complete spectrum of static one-dimensional solitons.
The stationarity property (harmonic with a non-traveling envelope) permits a reduction of the polariton equations 
to a real first-order ordinary differential equation.  
This allows symbolic integration of the polariton equations, resulting in exact analytical formulae for stationary polariton solitons.
For attractive exciton-exciton interactions we find two bands of bright solitons while for repulsive interactions
we find two bands of dark solitons.
In addition, a band of dark solitons with a discontinuous exciton field at the soliton center (discontinuous solitons) 
is found for attractive interactions and a band of discontinuous bright solitons with a nonzero background field 
is found for attractive interactions.
One-dimensional solitons have been shown to be realizable in a polariton waveguide
through detuning of the microcavity in the $x_2$ (transverse) direction \cite{EgorovSkryabinYulin2009}.

The system of two equations for the exciton and photon wavefunctions can, in general, not be reduced to the
Gross-Pitaevskii model for a single wavefunction representing polaritons.
A reduction is possible in certain regimes, and has been given in \cite[section~III]{KomineasShipmanVenakides2015}, where it is shown to apply at the left end of band 1.1 of dark solitons for $g>0$ (and, more generally, where $\zeta\approx\gamma^2$).  
Modeling the full range of solitons requires the system of two equations.  Specifically, for the bands of discontinuous solitons, one field vanishes where the other one is nonzero, and this phenomenon obviously lies outside the parameter regime of validity of the Gross-Pitaevskii model.

The six bands of solitons we discover are the static members of presumably much larger classes of solutions of the one-dimensional polariton system that include traveling and forced solitons.
For example, one-dimensional stable traveling bright solitons for $g>0$, sustained by an optical source, are reported in \cite{EgorovSkryabinYulin2009}.
Finding traveling soliton solutions of the polariton equations, even in one spatial dimension, is not a simple matter.
This is because the exciton and polariton wavefunctions in \eqref{traveling_solitons} are in general complex, resulting in a fourth-order system of real ODEs. 

Two-dimensional polaritons exhibit an abundance of interesting phenomena that promise some mathematical challenges.  
Unlike the nonlinear Schr\"odinger equation, the polariton system (\ref{2Dpolaritons}) with $\kappa_X=\kappa_C=0$ is not invariant under the Galilean transformation
\begin{equation}
  (\exciton(\mathbf{x},t),\photon(\mathbf{x},t)) \,\mapsto\, (\exciton(\mathbf{x}\!-\!\boldsymbol{\xi}t,\,t),\photon(\mathbf{x}\!-\!\boldsymbol{\xi}t,\,t))e^{i(\mathbf{k}\cdot\mathbf{x}-\omega t)},
\end{equation}
in which $\boldsymbol{\xi}={\textstyle\frac{\hbar}{m_C}}\mathbf{k}$ and $\omega = {\textstyle\frac{\hbar}{2m_C}}|\mathbf{k}|^2$.
The upper left entry of the matrix in (\ref{2Dpolaritons}) gains a transport term $-i\,\boldsymbol{\xi}\cdot\nabla\exciton$ and a frequency shift $\ox\mapsto \ox - \omega$.  The scaling transformation
\begin{equation}
  (\exciton(\mathbf{x},t),\photon(\mathbf{x},t)) \,\mapsto\, \lambda (\exciton(\lambda^2t,\lambda\mathbf{x}),\photon(\lambda^2t,\lambda\mathbf{x}))\,,
\end{equation}
which preserves the cubic nonlinear Schr\"odinger equation, effects a transformation of the polariton equations through a scaling of the frequency parameters,
\begin{equation}
  (\ox,\,\oc,\,\gamma) \,\mapsto\, \lambda^2(\ox,\,\oc,\,\gamma)\,.
\end{equation}
The result is a simple scaling by $\lambda^2$ of both the $\omega$ and the $\zeta$ axes in the depiction of the band structure of solitons in Fig.~\ref{fig:bands}.

\vspace{5ex}

\noindent
{\bfseries Acknowledgment.}
This work was partially supported by the European Union's FP7-REGPOT-2009-1 project 
``Archimedes Center for Modeling, Analysis and Computation'' (grant agreement n.\,245749) and by the US National Science  Foundation under grants NSF DMS-0707488 and NSF DMS-1211638.
We acknowledge discussions on the physics and experimental aspects of polariton condensates with P. Savvidis, G. Christmann, F. Marchetti.

\vspace{5ex}

\bibliography{KSV}

\end{document}